\documentclass[letterpaper,11pt]{article}
\usepackage[utf8]{inputenc}

\pdfoutput=1

\usepackage{amsmath, amsthm, amssymb}
\usepackage[linesnumbered,vlined]{algorithm2e}
\usepackage{mathtools}
\usepackage[numbers]{natbib}
\usepackage{comment} 
\usepackage{color-edits} 
\usepackage[most]{tcolorbox}
\usepackage{xfrac}
\usepackage{hyperref}
\usepackage{multirow}
\usepackage{caption}
\usepackage{bm}
\usepackage{newfloat}
\usepackage{enumitem}
\usepackage{xpatch}
\usepackage{soul}
\usepackage{bbm}

\usepackage[margin=1in]{geometry}

\allowdisplaybreaks

\definecolor{mygreen}{RGB}{20,120,60}

\title{Stochastic Matching with Few Queries: $(1-\varepsilon)$ Approximation\footnote{A version of this paper is to appear at STOC 2020.}}

\author{Soheil Behnezhad\thanks{Supported by a Google PhD Fellowship.} \and Mahsa Derakhshan \and MohammadTaghi Hajiaghayi}

\date{University of Maryland\\ \texttt{\{soheil,mahsa,hajiagha\}@cs.umd.edu}
}

\newcommand{\local}[0]{\ensuremath{\mathsf{LOCAL}}}

\newcommand{\E}[0]{\ensuremath{\mathbb{E}}}

\DeclareMathOperator{\Var}{Var}
\DeclareMathOperator{\Cov}{Cov}

\newcommand{\opt}[0]{\ensuremath{\textsc{opt}}}

\newcommand{\alg}[0]{\ensuremath{\mathsf{ALG}}\ref{alg:sampling}}

\newcommand{\findmatching}[2]{\ensuremath{\mathsf{FindMatching}_{#1}(#2)}}
\newcommand{\apxMIS}[1]{\ensuremath{\mathsf{ApproximateMIS}(#1)}}

\newcommand{\MM}[1]{\mathsf{MM}(#1)}

\DeclareMathOperator{\poly}{poly}
\DeclareMathOperator{\polylog}{polylog}

\DeclareMathOperator{\var}{Var}

\renewcommand{\epsilon}[0]{\ensuremath{\varepsilon}}

\let\originalleft\left
\let\originalright\right
\renewcommand{\left}{\mathopen{}\mathclose\bgroup\originalleft}
\renewcommand{\right}{\aftergroup\egroup\originalright}

\addauthor{sb}{blue}    
\addauthor{md}{red}    

\newtheorem{theorem}{Theorem}

\newtheorem{lemma}{Lemma}[section]

\newtheorem{corollary}[lemma]{Corollary}

\newtheorem{definition}[lemma]{Definition}
\newtheorem{claim}[lemma]{Claim}

\newtheorem{observation}[lemma]{Observation}
\newtheorem{remark}[lemma]{Remark}
\newtheorem{assumption}[lemma]{Assumption}

\makeatletter
\def\thm@space@setup{%
  \thm@preskip= 0.2cm
  \thm@postskip=\thm@preskip 
}
\makeatother

\definecolor{mygreen}{RGB}{20,125,20}
\definecolor{myred}{RGB}{200,20,20}
\definecolor{linkcolor}{RGB}{0,0,230}
\definecolor{mylightgray}{RGB}{230,230,230}
\definecolor{verylightgray}{RGB}{240,240,240}
\definecolor{commentcolor}{RGB}{120,120,120}

\SetCommentSty{mycommfont}

\newcommand{\smparagraph}[1]{
\par\addvspace{0.2cm}
\noindent \textbf{#1}
}

\hypersetup{
     colorlinks=true,
     citecolor= mygreen,
     linkcolor= myred,
     urlcolor= mygreen
}

\newcommand{\etal}[0]{\textit{et al.}}

\newcommand{\mc}[1]{\ensuremath{\mathcal{#1}}}

\newcounter{myalgctr}

\newenvironment{tbox}{
\par\addvspace{0.2cm}
\begin{tcolorbox}[width=\textwidth,
                  enhanced,
                  boxsep=2pt,
                  left=1pt,
                  right=1pt,
                  top=4pt,
                  boxrule=1pt,
                  arc=0pt,
                  colback=white,
                  colframe=black,
                  unbreakable
                  ]
}{
\end{tcolorbox}
}

\newenvironment{tboxh}{
\par\addvspace{0.2cm}
\begin{tcolorbox}[width=\textwidth,
                  enhanced,
                  boxsep=2pt,
                  left=1pt,
                  right=1pt,
                  top=4pt,
                  boxrule=1pt,
                  arc=0pt,
                  colback=white,
                  colframe=black,
                  unbreakable,
                  float=t
                  ]
}{
\end{tcolorbox}
}

\newcommand{\algcomment}[1]{{\color{gray} // #1 }}

\newcommand{\tboxhrule}[0]{\vspace{0.1cm} \hrule \vspace{0.2cm}}

\newenvironment{titledtbox}[1]{\begin{tbox}#1 \tboxhrule}{\end{tbox}}
\newenvironment{titledtboxh}[1]{\begin{tboxh}#1 \tboxhrule}{\end{tboxh}}

\newenvironment{tboxalg}[2][]{\refstepcounter{myalgctr}\begin{titledtbox}{\textbf{Algorithm \themyalgctr}#1\textbf{.} #2}}{\end{titledtbox}}

\newenvironment{tboxalgh}[2][]{\refstepcounter{myalgctr}\begin{titledtboxh}{\textbf{Algorithm \themyalgctr}#1\textbf{.} #2}}{\end{titledtboxh}}

\newenvironment{highlighttechnical}[0]{
\vspace{0.1cm}
\begin{tcolorbox}[width=\textwidth,
                  enhanced,
                  boxsep=2pt,
                  left=1pt,
                  right=1pt,
                  top=4pt,
                  boxrule=0.8pt,
                  arc=0pt,
                  colback=lightgray,
                  colframe=black,
                  unbreakable
                  ]
}{
\end{tcolorbox}
}

\newcommand{\restatedesc}[3]{
\par\addvspace{0.2cm}
\noindent\textbf{#1 {\normalfont (#2)}.} {\em #3}
\par\addvspace{0.2cm}
}

\newcommand{\restate}[2]{
\restatedesc{#1}{restated}{#2}
}

\begin{document}
\maketitle

\thispagestyle{empty}
\begin{abstract}
\setlength{\parskip}{0.3em}

	Suppose that we are given an arbitrary graph $G=(V, E)$ and know that each edge in $E$ is going to be {\em realized} independently with some probability $p$. The goal in the {\em stochastic matching} problem is to pick a sparse subgraph $Q$ of $G$ such that the realized edges in $Q$, in expectation, include a matching that is approximately as large as the maximum matching among the realized edges of $G$. The maximum  degree of $Q$ can depend on $p$, but not on the size of $G$.
	
	This problem has been subject to extensive studies over the years and the approximation factor has been improved from $0.5$ \cite{blumetal,AKL16} to $0.5001$ \cite{AKL17} to $0.6568$ \cite{soda19} and eventually to $2/3$ \cite{sosa19}. In this work, we analyze a natural sampling-based algorithm and show that it can obtain all the way up to $(1-\epsilon)$ approximation, for any constant $\epsilon > 0$.
	
	A key and of possible independent interest component of our analysis is an algorithm that constructs a matching on a {\em stochastic} graph, which among some other important properties, guarantees that each vertex is matched {\em independently} from the  vertices that are sufficiently far. This allows us to bypass a previously known barrier \cite{AKL16,AKL17} towards achieving $(1-\epsilon)$ approximation based on existence of dense Ruzsa-Szemerédi graphs.
\end{abstract}
\clearpage

{

\hypersetup{
     linkcolor= black
}

\thispagestyle{empty}
\tableofcontents{}
\clearpage
}

\setcounter{page}{1}

\section{Introduction}\label{sec:intro}
{

We study the following {\em stochastic matching} problem. An arbitrary graph $G=(V, E)$ is given, then each edge $e \in E$ is  retained (or to be consistent with the literature {\em realized}) independently with some given probability $p \in (0, 1]$. The goal is to pick a subgraph $Q$ of $G$ without knowing the edge realizations such that:
\begin{enumerate}[itemsep=0pt,topsep=5pt]
	\item The expected size of the maximum matching among the realized edges of $Q$ approximates the expected size of the maximum matching among the realized edges in $G$.
	\item The maximum degree in $Q$ is bounded by a function that may depend on $p^{-1}$ but must be independent of the size of $G$.\footnote{In this paper, we solve a generalization of this problem where each edge $e$ has its own  realization probability $p_e$ and the degree of $Q$ can be proportional to $p = \min_e p_e$. See Section~\ref{sec:prelim} for the formal setting.}
\end{enumerate}
It would be useful to think of $p$ as some constant whereas $n := |V| \to \infty$. Then the second condition translates to $Q$ having $O(1)$ maximum degree. In other words, the subgraph $Q$ should provide a good approximation while having $O(n)$ edges, in contrast to $G$ which may have up to $\Omega(n^2)$ edges.

\smparagraph{Applications.} The setting is mainly motivated by applications in which the process of determining an edge realization (referred to as {\em querying} the edge) is considered time consuming or expensive. For such applications, one can instead of querying every edge of $G$, only query the edges of its much sparser subgraph $Q$ and still find a large realized matching in $G$. Kidney exchange and online labor markets are major examples of such applications. For more details on the role of the stochastic matching problem in these applications, see \cite{arXivblumetal,blumetal,AKL16,AKL17,BR18} (particularly \cite[Section~1.2]{arXivblumetal}) for kidney exchange and \cite{BR18,soda19,sagt19} for online labor markets. Another natural application of the model is that this subgraph $Q$ can be used as a {\em matching sparsifier} for $G$ which approximately preserves its maximum matching size under random edge failures \cite{sosa19}.

\smparagraph{Related work.} The problem has received significant attention \cite{blumetal,AKL16,AKL17,YM18,BR18,soda19,sosa19,sagt19} after the pioneering work of Blum~\etal{}~\cite{blumetal} who proved that it admits a $(\frac{1}{2}-\epsilon)$-approximation.  Earlier follow-up works revolved around the prevalent half-approximation barrier until it was first broken by Assadi~\etal{}~\cite{AKL16}. This was followed by a $0.6568$-approximation by Behnezhad~\etal{}~\cite{soda19} and eventually a $(\frac{2}{3}-\epsilon)$-approximation by Assadi and Bernstein \cite{sosa19} which is the state-of-the-art. See also \cite{YM18,BR18,soda19,YM19} for various natural generalizations of the problem.

\smparagraph{Our result.} In this work, we improve the approximation-factor all the way up to $(1-\epsilon)$:

\begin{highlighttechnical}
	\begin{theorem}\label{thm:main}
		For any $\epsilon > 0$, there is an algorithm that picks an $O_{\epsilon, p}(1)$-degree subgraph $Q$ of $G$ such that the expected size of the maximum realized matching in $Q$ is at least $(1-\epsilon)$ times the expected size of the maximum realized matching in $G$.
	\end{theorem}
\end{highlighttechnical}

To get a $(1-\epsilon)$-approximation, the dependence of the maximum degree of $Q$ on both $\epsilon$ and $p$ is necessary. Particularly, a simple lower bound shows that even when $G$ is a clique, to avoid too many singleton vertices in a realization of $Q$, the maximum degree in $Q$ must be $\Omega(\frac{\ln \epsilon^{-1}}{p})$ \cite{AKL16}. The same lower bound also shows that a $(1-o(1))$ approximation is not achievable unless the maximum degree of $Q$ is $\omega(1)$, meaning that our approximation-factor is essentially the best one can hope for.

\begin{remark}
	The $O_{\epsilon, p}(1)$ term in Theorem~\ref{thm:main} is in the order 
	$\exp\left(\exp\left(\exp\left(O\left(\epsilon^{-1} \right)\right) \times \log \log p^{-1}\right)\right)$. We do not believe this dependence is optimal and leave it as an open problem to improve it. Particularly, we conjecture that the same algorithm that is analyzed in this work (see Algorithm~\ref{alg:sampling}) should obtain up to  $(1-\epsilon)$-approximation even by picking only a $\poly(1/\epsilon p)$-degree subgraph.
\end{remark}

\smparagraph{The algorithm.} Many different constructions of $Q$ have been studied in the literature. A well-studied algorithm first considered by Blum~\etal{}~\cite{blumetal} which was further analyzed (module minor differences and generalizations) in the subsequent works of \cite{AKL16,AKL17,BR18,YM18,YM19} is as follows: Iteratively pick a maximum matching $M_i$ from $G$, remove its edges, and finally let $Q = M_1 \cup \ldots \cup M_R$ for some parameter $R$ that controls the maximum degree in $Q$. Despite the positive results proved for this algorithm, it was already shown in \cite{blumetal} that its approximation-factor is not better than $5/6$. Thus to obtain $(1-\epsilon)$-approximation, one has to use a different algorithm.

We focus on an algorithm proposed previously by Behnezhad~\etal{}~\cite{soda19}, which they proved obtains at least a $0.6568$-approximation. The algorithm is equally simple, but subtly different: Draw $R$ independent realizations $\mc{G}_1, \ldots, \mc{G}_R$ of $G$ and let $Q = \MM{\mc{G}_1} \cup \ldots \cup \MM{\mc{G}_R}$ where $\MM{\mc{G}_i}$ is a maximum matching of $\mc{G}_i$. Our main result is obtained via providing a different analysis of this algorithm. Within the next two paragraphs, we discuss how our analysis differs substantially from the previous approaches and in particular from the analysis of \cite{soda19}.

\smparagraph{The analysis and the Ruzsa-Szemerédi barrier.} A major barrier to overcome in order to prove existence of a $(1-\epsilon)$-approximate subgraph was already discussed in the work of Assadi, Khanna, and Li \cite[Section~6]{AKL16} based on Ruzsa-Szemerédi graphs \cite{ruzsa1978triple,DBLP:conf/stoc/FischerLNRRS02,DBLP:conf/soda/GoelKK12,DBLP:conf/stoc/AlonMS12} which we henceforth call the ``RS-barrier''. Consider an extension of the stochastic matching setting where the realization of edges in a single a-priori known matching $M$ of $G$ can be correlated while other edges are still realized independently. An implication of the RS-barrier is that in this extended model, no algorithm can obtain $(1-\epsilon)$-approximation (or even beat $\frac{2}{3}$-approximation\footnote{The original proof of \cite{AKL16} rules out  $>\frac{6}{7}$-approximation. A similar instance can rule out $\frac{2}{3}$-approximation using a more efficient construction of RS-graphs  \cite{DBLP:conf/soda/GoelKK12} and allowing a subset of edges of $G$ to have realization probability 1.}) unless $Q$ has maximum degree  $n^{\Omega(1/\log\log n)} = \omega(\polylog n)$. Put differently, this proves that in order to beat $\frac{2}{3}$-approximation, the analysis has to use the fact that {\em every} edge around a vertex is realized independently. This explains why the previous arguments were short of bypassing $\frac{2}{3}$-approximation: They can all (to our knowledge) be adapted to tolerate  adversarial realization of one edge per vertex.

\smparagraph{``Vertex-independent matchings'' to the rescue.} We overview our analysis soon in Section~\ref{sec:techniques}. However, here we briefly mention our key analytical tool in bypassing the RS-barrier. It is an algorithm (Lemma~\ref{lem:independentmatching}) for constructing a matching $Z$ on the realized {\em crucial} edges (roughly, an edge is crucial if it has a sufficiently high probability of being part of an optimal realized matching). The algorithm constructs $Z$ such that among some other useful properties, it guarantees that each vertex is matched independently from all but $O(1)$ other vertices. Here the independence is with regards to both the randomization of the algorithm in constructing $Z$, and also importantly  \underline{the edge realizations of $G$}. This independence property is the key that separates the stochastic matching model from the extended model of the RS-barrier: Due to the added correlations in the edge realizations, such vertex-independent matchings essentially do not exist in the model of the RS-barrier. Using this independence, we show that $Z$ can be well-augmented by the rest of the realized edges in $Q$. See Section~\ref{sec:techniques} for a more detailed overview of our analysis and how the independence property helps.

Our method of bypassing the RS-barrier via vertex-independent matchings sheds more light on the limitations imposed by Ruzsa-Szemerédi type graphs. These graphs are known to be notoriously hard examples in various other areas such as property testing, streaming algorithms, communication complexity, and additive combinatorics among others \cite{DBLP:conf/soda/Kapralov13,DBLP:conf/soda/GoelKK12,DBLP:conf/stoc/AlonMS12,ruzsa1978triple,DBLP:conf/stoc/FischerLNRRS02,gowers2001some}. As such, we believe that this method may find applications beyond the stochastic matching problem.

\smparagraph{Organization of the paper.} In Section~\ref{sec:techniques} we provide an informal overview of our  analysis. In Section~\ref{sec:prelim} we formally state the problem and the notations used throughout the paper. In Section~\ref{sec:analysissetup} we describe the algorithm and basic definitions that we will use throughout the analysis. In Section~\ref{sec:analysisviavertexindependent} we prove how the vertex-independent matching lemma leads to a $(1-\epsilon)$-approximation and in Section~\ref{sec:independentmatching}, we prove the vertex-independent matching lemma. Finally, Section~\ref{sec:proofs} contains the proofs of (less important) statements that are deferred.

}
\section{Our Techniques}\label{sec:techniques}

As previously described, we consider the following algorithm for constructing subgraph $Q$ (see also Algorithm~\ref{alg:sampling}): Draw $R$ realizations $\mc{G}_1, \ldots, \mc{G}_R$ of graph $G$, then pick a matching $\MM{\mc{G}_i}$ from each realization, and finally set $Q = \MM{\mc{G}_1} \cup \ldots \cup \MM{\mc{G}_R}$. In this section, we give an informal overview of our analysis for this algorithm. 

Note that these realizations $\mc{G}_i$ are part of the randomization of the algorithm and may be very different from the actual realization $\mc{G}$ of $G$. In fact, in expectation, only $p$ fraction of the edges of each matching $\MM{\mc{G}_i}$ are realized in $\mc{G}$. Thus, we have to argue that the realized edges of these matchings can be used to augment each other and form a large matching in the realized subgraph $\mc{Q}$ of $Q$. In order to do this, we will give a ``procedure'' to construct a matching in $\mc{Q}$. To get a handle on the dependencies involved, the procedure carefully decides how the realization of edges in $Q$ are revealed and which are chosen to be in the matching. We emphasize that this procedure is merely an analytical tool for analyzing the approximation-factor. Thus, no matter how intricate it is, the algorithm for constructing $Q$ remains to be the simple Algorithm~\ref{alg:sampling} described above.

\smparagraph{A crucial/non-crucial decomposition.} Similar to  \cite{soda19} (and also implicitly \cite{AKL17}), we consider a partitioning of the edges of $G$ into what we call {\em crucial} and {\em non-crucial} edges. For each edge $e$, define $q_e := \Pr[e \in \MM{\mc{G}}]$ where $\MM{\cdot}$ is the same matching algorithm used to construct $Q$. We further assume that $\MM{\cdot}$ is deterministic, so the probability is taken only over the realization $\mc{G}$. For two thresholds $0 < \tau_- < \tau_+ < 1$ that we fix later, we define:
\begin{itemize}[itemsep=0pt]
	\item The crucial edges as $C := \{ e \in E \mid q_e \geq \tau_+\}$.
	\item The non-crucial edges as $N = \{ e \in E \mid q_e \leq \tau_-\}$.
\end{itemize}
Note that in the decomposition above edges $e$ with $q_e \in (\tau_-, \tau_+)$ are neither crucial nor non-crucial. We will essentially ``ignore'' these edges in the analysis but ensure that we choose $\tau_-$ and $\tau_+$ such that there are few ignored edges. 

In our procedure to construct a matching on $\mc{Q}$, we treat crucial and non-crucial edges differently. We start with the crucial edges and (in Lemma~\ref{lem:independentmatching}) construct a matching $Z$ on them whose expected size is (almost) as large as the expected number of crucial edges in the optimal maximum realized matching of $G$. We then show that this matching $Z$ can be augmented via the non-crucial edges to eventually form a matching whose expected size is arbitrarily close to $\opt := \E[|\MM{\mc{G}}|]$.

\smparagraph{The procedure for crucial edges.} In addition to the lower bound on the expected size of $Z$, we make sure that no vertex tends to be ``over-matched'' in $Z$. More formally, the probability of any vertex $v$ being matched in $Z$ should not be larger than the probability that $v$ is matched via a crucial edge in $\MM{\mc{G}}$. Both of these conditions can actually be satisfied by a very simple randomized procedure: Reveal the whole realization $\mc{C}$ of $C$, also draw a random realization $\mc{N}'$ of the non-crucial edges, and let $Z$ be the crucial edges in matching $\MM{\mc{C} \cup \mc{N}'}$. 

Unfortunately, the matching constructed via the above-mentioned procedure is hard to augment via the non-crucial edges as we have no control over the correlations. To get around this, we need an extra  ``independence'' property. Let $X_v$ be the indicator of the event that vertex $v$ is matched in $Z$. The independence property requires random variables $X_{v_1}, X_{v_2}, \ldots, X_{v_n}$ to be (almost) independent where $\{v_1, \ldots, v_n\}$ is the vertex-set of $G$. Clearly, perfect independence cannot be achieved: Given the event that a vertex $v$ is matched in $Z$, we derive that at least one of its neighbors in $C$ is also matched. What we prove can be achieved, though, is that each $X_{v}$ is independent from $X_u$ of vertices $u$ outside a small local neighborhood of $v$ in graph $C$. (See Lemma~\ref{lem:independentmatching} part~4 for the formal statement.)

In order to satisfy the independence property described above, we will not reveal the whole realization $\mc{C}$ outright and then construct $Z$ based on it as it was done in the simple procedure described above. Instead, we present a different algorithm (Algorithm~\ref{alg:crucial}) for constructing this matching $Z$. To prove the independence property, we show that this algorithm can be simulated locally. In other words, for each vertex $v$, the value of $X_v$ can be determined uniquely by having the realization of edges in a small local neighborhood of $v$. Thus, if two vertices $u$ and $v$ are sufficiently far from each other in graph $C$, then $X_v$ and $X_u$ would be independent. 

\smparagraph{Augmenting $Z$ via non-crucial edges.} We noted above that $\E[|Z|]$ is (almost) as large as the expected number of crucial edges in $\MM{\mc{G}}$. Therefore, in order to construct a matching of $\mc{Q}$ with expected size arbitrarily close to $\opt$, we have to augment $Z$ via the non-crucial edges. To do this, we only use non-crucial edges $\{u, v\}$ in $Q$ such that $X_u$ and $X_v$ are independent. Describing how exactly we construct the matching on these non-crucial edges requires a number of definitions which we give in Section~\ref{sec:constructfractional}. However, to convey the key intuition, here we only mention how and why the independence of $X_u$ and $X_v$ plays an important role in using a non-crucial edge $e = \{u, v\}$ to augment $Z$. Suppose that $\Pr[X_u] = \Pr[X_v] = 1/2$. Note that it is only when {\em both} $u$ and $v$ are unmatched in $Z$ that we can use edge $e$ to augment $Z$. If $X_u$ and $X_v$ are independent, there is a relatively large probability $(1-\Pr[X_u])(1-\Pr[X_v]) = \frac{1}{4}$ that this occurs. However, if $X_u$ and $X_v$ can be correlated, it may be the case that with probability half  $X_u = 1$ and $X_v = 0$, and with probability half $X_u = 0$ and $X_v = 1$. In this case, the probability of both $u$ and $v$ being unmatched in $Z$ would be zero and thus we would never be able to use $e$ to augment $Z$. We remark that this is precisely the type of correlation introduced in the RS-barrier of \cite{AKL16} which the independence property allows us to bypass.

\section{Preliminaries}\label{sec:prelim}

\paragraph{General notations.} We denote the maximum matching size of any graph $G$ by $\mu(G)$. For a matching $M$, we use $V(M)$ to denote the set of vertices matched in $M$. For any two nodes $u$ and $v$ in a graph $G$, we use $d_G(u, v)$ to denote their distance, i.e. the number of edges in their shortest path. Furthermore, the distance $d_G(u, e)$ between an edge $e$ and a node $u$ is the minimum distance between an endpoint of $e$ and $u$. We use $\mathbbm{1}(A)$ as the {\em indicator} of an event $A$, i.e. $\mathbbm{1}(A) = 1$ if event $A$ occurs and $\mathbbm{1}(A) = 0$ otherwise. Also, we may use $[k] := \{1, 2, \ldots, k\}$ for any integer $k \geq 1$.

Throughout the paper, we define various functions of form $x : E \to [0, 1]$ that map each edge $e \in E$ to a real number in $[0, 1]$. Having such function $x$, for any vertex $v$ we define $x_v := \sum_{e \ni v} x_e$, for any edge subset $F$ we define $x(F) := \sum_{e \in F} q_e$, and for any vertex subset $U$ we define $x(U) := \sum_{e=\{u, v\}: u, v \in U} x_e$. We also denote $|x| = \sum_e x_e$.

\smparagraph{The setting.} We consider a generalized variant of the standard stochastic matching problem studied in the literature where each edge $e$ has a realization probability $p_e$ that may be different from that of other edges. We then let $p = \min_e p_e$, which is the parameter the degree of subgraph $Q$ can depend on. This generalization will actually help in solving the original model of the literature defined in Section~\ref{sec:intro} which coincides with the case where $p_e = p$ for every edge $e$.

We denote realizations by script font; for instance, we use $\mc{G}=(V, \mc{E})$ to denote the realized subgraph of the input graph $G$, which includes each edge $e$ independently with probability $p_e$. Similarly, we use $\mc{Q}$ to denote the realized subgraph of $Q$. The same notation also naturally extends to denote realization of other subgraphs of $G$ that we may later define.

As discussed in Section~\ref{sec:intro}, the goal is to pick a sparse subgraph $Q$ of $G$ such that the ratio $\E[\mu(\mc{Q})]/\E[\mu(\mc{G})]$, known as the approximation-factor, is large. Here  the expectations are taken over the realizations $\mc{Q}$ and $\mc{G}$, and possibly the randomization of the algorithm in constructing subgraph $Q$. For brevity, we use $\opt$ to denote  $\E[\mu(\mc{G})]$. Note that $\opt$ is just a number.

We note that the {\em expected} approximation-factor defined above can automatically be turned into {\em high-probability} due to a simple concentration bound. See Appendix~\ref{sec:concentration}.

\section{The Algorithm and Basic Definitions}\label{sec:analysissetup}

The algorithm that we analyze is formally stated as Algorithm~\ref{alg:sampling}.

\begin{tboxalg}[ (\cite{soda19})]{A sampling-based non-adaptive algorithm for stochastic matching.}\label{alg:sampling}
	\textbf{Parameter:} $R$, which controls the maximum degree of $Q$.
	\vspace{0.1cm}
	
	Take $R$ realizations $\mc{G}_1, \ldots, \mc{G}_R$ of $G$ independently where each realization $\mc{G}_i$ includes each edge $e$ independently with probability $p_e$. Return subgraph $Q = \MM{\mc{G}_1} \cup \ldots \cup \MM{\mc{G}_R}$.
\end{tboxalg}

In the algorithm above, $\MM{\mc{G}_i}$ returns a maximum matching of $\mc{G}_i$. It will be convenient for the analysis to assume $\MM{\cdot}$ is a deterministic maximum matching algorithm.

In order to analyze Algorithm~\ref{alg:sampling}, we will make the following assumption which will simplify many of our arguments.

\begin{assumption}\label{ass:optlarge}
	$\opt \geq 0.1\epsilon n$.
\end{assumption}

Assumption~\ref{ass:optlarge} comes w.l.o.g. due to a reduction of Assadi~\etal{}~\cite{AKL16}. The reduction is roughly as follows: If $n \gg \opt$, randomly put nodes of $G$ into $O(\frac{\opt}{\epsilon})$ buckets and contract the nodes within each bucket. The resulting graph will have only $O(\frac{\opt}{\epsilon})$ nodes but its expected maximum realized matching will be as large as $(1-O(\epsilon))\opt$. Solving this modified graph will then solve the original graph $G$ as well. We provide further details in Appendix~\ref{app:optlarge} and note that for the reduction to work, it is important that our algorithm can handle different edge realization probabilities.

\subsection{A Crucial/Non-crucial Decomposition}

For each edge $e$ define $q_e := \Pr[e \in \MM{\mc{G}}]$ where $\MM{\cdot}$ is the same matching algorithm used in Algorithm~\ref{alg:sampling}. Since we assumed $\MM{\cdot}$ is deterministic, the probability is taken only over the randomization of the realization $\mc{G}$. Having this definition, for any vertex $v$ we denote $q_v := \sum_{e \ni v} q_e$ and for any subset $E' \subseteq E$ denote $q(E') := \sum_{e \in E'} q_e$. The following statements immediately follow from the definition:

\begin{observation}
	$q(E) = \opt$.
\end{observation}
\begin{observation}\label{obs:qv}
	For any vertex $v$, $q_v$ denotes the probability that $v$ is matched in $\MM{\mc{G}}$.
\end{observation}

We will fix two thresholds $0 < \tau_- < \tau_+ < 1$ that both depend only on $\epsilon$ and $p$. Next, for any edge $e$, we say $e$ is {\em crucial} if $q_e \geq \tau_+$,  {\em non-crucial} if $q_e \leq \tau_-$, and {\em ignored} if $q_e \in (\tau_-, \tau_+)$. We denote the crucial edges by $C := \{ e \in E \mid \text{$e$ is crucial} \}$, and the non-crucial edges by $N := \{ e \in E \mid \text{$e$ is non-crucial} \}$. Furthermore, we denote their realizations by $\mc{C} := C \cap \mc{E}$ and $\mc{N} := N \cap \mc{E}$. When confusion is impossible, we may use $C$ to denote graph $(V, C)$ instead of merely the edge-subset. The same also naturally generalizes to $N$, $\mc{C}$, and $\mc{N}$. We will further use $\Delta_C$ to denote the maximum degree in graph $C$. Moreover, for any vertex $v$ we use $c_v$ (resp. $n_v$) to  denote the probability that $v$ is matched via a crucial (resp. non-crucial) edge in $\MM{\mc{G}}$.

\begin{observation}\label{obs:crucialdegree}
	$\Delta_C \leq 1/\tau_+$.
\end{observation}
\begin{proof}
	Each edge $e \in C$ has $q_e \geq \tau_+$ by definition. Thus, if there is a vertex $v$ of degree larger than $1/\tau_+$ in $C$, then it should hold that $q_v > 1/\tau_+ \times \tau_+ = 1$ which contradicts Observation~\ref{obs:qv}.
\end{proof}

\subsection{Setting the Thresholds $\tau_-$ and $\tau_+$}

To describe how we set the values of $\tau_-$ and $\tau_+$, we state a lemma that we prove in Section~\ref{sec:proofs}.

\begin{lemma}\label{lem:gap}
	Fix any arbitrary function $f(x)$ such that $0 < f(x) < x$ for any $0 < x < 1$. There is a choice of $0 < \tau_- < \tau_+ < 1$ such that: (1) $\tau_- = f(\tau_+)$. (2) $q(N) + q(C) \geq (1-\epsilon)\opt$. (3) Both $\tau_-$ and $\tau_+$ depend only on $\epsilon$ and $p$. And finally, (4) $\tau_+ \leq (\epsilon p)^{50}$.
\end{lemma}

The lemma above essentially shows that we can have any desirably large gap between $\tau_+$ and $\tau_-$ and still ensure that $q(N)+q(C) \geq (1-\epsilon)\opt$. That is, the ignored edges in expectation constitute at most $\epsilon \opt$ edges of  $\MM{\mc{G}}$. While this may sound counter-intuitive, it follows roughly speaking from the fact that by iteratively reducing the threshold $\tau_+$ by a sufficient amount, all the previously ignored edges become crucial. Thus it cannot continue to hold that there are still a significant mass of the matching on the ignored edges after sufficiently many iterations. See Section~\ref{sec:proofs} for the proof.

Having Lemma~\ref{lem:gap}, we set our thresholds and the parameter $R$ of Algorithm~\ref{alg:sampling} as follows:

\begin{highlighttechnical}
	\textbf{Setting $\tau_-, \tau_+,$ and $R$:}
	
	\vspace{0.4cm}
	
	Define function $f(x) := x^{10g(x)}$ where $g(x) := \epsilon^{-20}\log\frac{1}{x}$.
	
	\smallskip
	
	We plug this function $f$ into Lemma~\ref{lem:gap} and define $\tau_-$ and $\tau_+$ accordingly. We also set $R = \frac{1}{2\tau_-}$.
\end{highlighttechnical}

Note that function $f$ as defined above satisfies $0 < f(x) < x$ for any $0 < x < 1$ since clearly $g(x) \geq 1$ so long as $0 < x < 1$. Therefore, we can indeed plug $f$ into Lemma~\ref{lem:gap}. This results in the following properties:

\begin{corollary}\label{cor:thresholds}
	It holds that: (1) $\tau_- = (\tau_+)^{10g}$ where $g = \epsilon^{-20}\log\frac{1}{\tau_+}$. (2) $q(N) + q(C) \geq (1-\epsilon)\opt$. (3) Both $\tau_-$ and $\tau_+$ depend only on $\epsilon$ and $p$ and thus $R=O_{\epsilon, p}(1)$. (4) $\tau_- < \tau_+ \leq (\epsilon p)^{50}$.
\end{corollary}

The next lemma shows that $R$ is set such that Algorithm~\ref{alg:sampling}  samples (almost) all crucial edges.

\begin{observation}\label{obs:samplealmostallcrucial}
	For every edge $e \in C$, $\Pr[e \in Q] \geq 1-\epsilon$.
\end{observation}
\begin{proof}
	Note that $e \in Q$ if there is at least one $i \in [R]$ where $e \in \MM{\mc{G}_i}$. The probability that $e \in \MM{\mc{G}_i}$ for any fixed $i$ is precisely $q_e$. Since realizations $\mc{G}_1, \ldots, \mc{G}_R$ are independent, it holds that $\Pr[e \not\in Q] = (1-q_e)^{R}$. On the other hand $q_e \geq \tau_+$ since $e$ is crucial. Also $R = \frac{1}{2\tau_-} > \ln \epsilon^{-1}/\tau_+$ where the latter inequality follows easily from Corrolary~\ref{cor:thresholds} part (1). Combining all of these gives: 
	$$
	\Pr[e \not\in Q] = (1-q_e)^R < (1-\tau_+)^{\ln \epsilon^{-1}/\tau_+} < e^{-\ln \epsilon^{-1}} = \epsilon.
	$$
	Therefore indeed $\Pr[e \in Q] \geq 1-\epsilon$.
\end{proof}

\subsection{The Vertex-Independent Matching Lemma}

As discussed before, a key technical contribution of this work that allows getting an arbitrary good approximation-factor is a ``vertex-independent matching'' lemma that we state here. The proof of this lemma is involved and thus we defer it to Section~\ref{sec:independentmatching}. In Section~\ref{sec:analysisviavertexindependent}, we show how Lemma~\ref{lem:independentmatching} can be used to analyze Algorithm~\ref{alg:sampling} and prove Theorem~\ref{thm:main}.

\newcommand{\independentmatching}[0]{There is a randomized algorithm that constructs an integral matching $Z$ of $\mc{C}$ (the realized subgraph of $C$) such that defining $X_v$ as the indicator random variable for $v \in V(Z)$, we get:
	\begin{enumerate}[itemsep=0.2pt,topsep=5pt]
		\item $\E[|Z|] \geq q(C) - 30\epsilon\opt$.
		\item For every vertex $v$, $\Pr[X_v] \leq \max\{c_v - \epsilon^2, 0\}$, where recall that $c_v$ is the probability that vertex $v$ is matched via a crucial edge in $\MM{\mc{G}}$.
		\item The matching $Z$ is independent of the realization of non-crucial edges in $G$.
		\item Let $\lambda := \epsilon^{-20}\log\Delta_C$. For every $k$ and every $\{v_1, v_2, \ldots, v_k \} \subseteq V$ such that $d_C(v_i, v_j) \geq \lambda$ for all $v_i \not= v_j$, random variables $X_{v_1}, \ldots, X_{v_k}$ are independent.
	\end{enumerate}
We emphasize that $\E[|Z|]$ and $X_v$ are both defined with respect to the randomizations in both the realization of $C$, and the randomization of the algorithm in constructing $Z$.
	}
\begin{lemma}[Vertex-Independent Matching Lemma]\label{lem:independentmatching}
	\independentmatching{}
\end{lemma}

\begin{observation}\label{obs:ggtlambda}
	Let $g$ be as defined in Corollary~\ref{cor:thresholds} and $\lambda$ be as defined in Lemma~\ref{lem:independentmatching}. Then it holds that $g \geq \lambda$.
\end{observation}
\begin{proof}
	Since $\lambda = \epsilon^{-20}\log \Delta_C$ by definition and $\Delta_C \leq 1/\tau_+$ by Observation~\ref{obs:crucialdegree}, we get that $\lambda \leq \epsilon^{-20}\log \frac{1}{\tau_+}$. On the other hand $g = \epsilon^{-20}\log \frac{1}{\tau_+}$. Therefore, $g \geq \lambda$.
\end{proof}
\section{The Analysis via the Vertex-Independent Matching Lemma}\label{sec:analysisviavertexindependent}

In this section, given correctness of Lemma~\ref{lem:independentmatching}, we prove Theorem~\ref{thm:main}. In what follows we give the outline of the proof by referring to the needed lemmas that will be proved in subsequent Sections~\ref{sec:constructfractional}, \ref{sec:validityofx}, \ref{sec:analysisfractional}, and \ref{sec:turntofracmatching}.

\smparagraph{Proof Outline for Theorem~\ref{thm:main}.} 
	Let $Q$ be the output of by Algorithm~\ref{alg:sampling} where parameter $R$ is set as described above. We show that one can construct a matching of expected size at least $(1-56\epsilon)\opt$ on the realized subgraph $\mc{Q}$ of $Q$. This implies that $\E[\mu(\mc{Q})] \geq (1-56\epsilon)\opt = (1-56\epsilon)\E[\mu(\mc{G})]$. In other words, this proves that the approximation-factor of the algorithm is at least $(1-56\epsilon)$. (Note this is equivalent to $(1-\epsilon)$ approximation  since one can choose $\epsilon$ to be any desirably small constant.)
	
	In order to construct a matching of expected size at least $(1-56\epsilon)\opt$ on $\mc{Q}$, we first describe how to construct an ``expected fractional matching'' (see Definition~\ref{def:expfractional}) $x$ on $\mc{Q}$ in Sections~\ref{sec:constructfractional}, \ref{sec:validityofx}, and \ref{sec:analysisfractional}. Later on, we show in Section~\ref{sec:turntofracmatching} how to turn $x$ into a  fractional matching $y$ on $\mc{Q}$ such that $\E[|y|] \geq (1-55\epsilon)\opt$ (see Lemma~\ref{lem:ylarge}). Finally, to turn $y$ into an {\em integral} matching, we show (Observation~\ref{obs:yblossom}) that the so called ``blossom inequalities'' of size up to $1/\epsilon$ also hold for $y$. That is, we show that for all vertex subsets $U \subseteq V$ with $|U| \leq 1/\epsilon$, we have $y(U) \leq \lfloor \frac{|U|}{2}\rfloor$. By Edmond's celebrated theorem \cite{edmonds1965maximum,schrijver2003combinatorial} on the matching polytope, this means that there is an integral matching of size at least $\frac{1}{1+\epsilon}|y| \geq (1-\epsilon)|y|$ in \mc{Q}. As described, $\E[|y|] \geq (1-55\epsilon)\opt$, thus indeed $\E[\mu(\mc{Q})] \geq (1-\epsilon)(1-55\epsilon)\opt \geq (1-56\epsilon)\opt$ as desired.
 
\subsection{Construction of an Expected Fractional Matching $x$ on $\mc{Q}$}\label{sec:constructfractional}

In this section, we describe an algorithm that constructs an ``expected fractional matching'' $x$ on $\mc{Q}$ as defined below.

\begin{definition}\label{def:expfractional}
	Let $\mathcal{A}$ be a random process that assigns a fractional value $x_e \in [0, 1]$ to each edge $e$ of a graph $G(V, E)$. We say $x$ is an expected fractional matching if:
	\begin{enumerate}[itemsep=0pt,topsep=5pt]
		\item For each vertex $v$, defining $x_v := \sum_{e \ni v} x_e$ we have $\E[x_v] \leq 1$.
		\item For all subsets $U \subseteq V$ with $|U|\leq 1/\epsilon$, $x(U) \leq \lfloor \frac{|U|}{2} \rfloor$ with probability 1.
	\end{enumerate}
\end{definition}

We emphasize that the definition only requires $\E[x_v] \leq 1$, thus depending on the coin tosses of the process, it may occur that $x_v > 1$, violating the constraints of a normal fractional matching. We will later argue that in our construction, the values of $x_v$'s are sufficiently concentrated around their mean and thus we can turn our expected fractional matching to an actual fractional matching of (almost) the same size.

As described before, we construct an expected fractional matching $x$ on the edges of graph $\mc{Q}$. Note that here the graph $\mc{Q}$ itself is also stochastic. In the construction, we treat crucial and non-crucial edges completely differently.

\smparagraph{Crucial edges.} On the crucial edges, we first construct an integral matching $Z$ using the algorithm of Lemma~\ref{lem:independentmatching}. Once we have $Z$, we define $x$ on crucial edges as follows.

\begin{highlighttechnical}
\vspace{-0.4cm}
\begin{flalign}\label{eq:crucialxe}
	\text{For every crucial edge $e$, } \qquad\qquad x_e := \begin{cases}
    1,& \text{if $e \in Z$ and $e \in Q$,}\\
    0,              & \text{otherwise}.
    \end{cases}
\end{flalign}
\end{highlighttechnical}

Note from Observation~\ref{obs:samplealmostallcrucial} that each crucial edges belong to $Q$ with probability at least $1-\epsilon$. Therefore the construction above (roughly speaking) sets $x_e = 1$ for most of the edges $e$ in $Z$.

\smparagraph{Non-crucial edges.} For defining $x$ on the non-crucial edges, we start with a number of useful definitions. For any edge $e$, define $t_e$ to be the number of matchings $\MM{\mc{G}_1}, \ldots, \MM{\mc{G}_R}$ that include $e$. Then based on that, define 
\begin{equation}\label{eq:deff}
	f_e := \begin{cases}
		\frac{t_e}{R}, & \text{if $\frac{t_e}{R} \leq \frac{1}{\sqrt{\epsilon R}}$ and $e$ is non-crucial,}\\
		0, & \text{otherwise.}
	\end{cases}
\end{equation}
Note that $f_e$ is a random variable of only the randomization of Algorithm~\ref{alg:sampling}, i.e. it is independent of the realization. Also note that $f_e$ is desirably non-zero only on the edges that belong to graph $Q$. Having defined $f_e$, we define $x_e$ on the non-crucial edges as follows.
\begin{highlighttechnical}
For every non-crucial edge $e$, define
\begin{flalign}\label{eq:noncrucialxe}
	x_e = \begin{cases}
    \frac{f_e}{p_e(1-\Pr[X_v])(1-\Pr[X_u])}, & \text{if $e$ is realized, $u, v \not\in V(Z)$, and $d_C(u, v) \geq \lambda$,}\\
    0,              & \text{otherwise}.
    \end{cases}
\end{flalign}
\end{highlighttechnical}

We note that $\lambda$ in the definition above is the number defined in Lemma~\ref{lem:independentmatching} and that $X_v$ is the indicator random variable for the event $v \in V(Z)$.

Before concluding this section, let $f_v := \sum_{e \in N : v \in e} f_e$ for each vertex $v$. We note the following properties of $f$, which can be derived directly from the definition above. The proof is given in Section~\ref{sec:proofs}.

\begin{claim}\label{cl:frange}
	It holds that:
	\begin{enumerate}
		\item For every non-crucial edge $e$, $\E[f_e] \leq q_e$.
		\item For every non-crucial edge $e$, $\E[f_e] \geq (1-\epsilon)q_e$.
		\item For every vertex $v$, it always holds that $\sum_{e \ni v} f_e \leq 1$.
		\item For every vertex $v$, $\Pr[f_v > n_v + 0.1\epsilon] \leq (\epsilon p)^{10}$, where recall that $n_v$ is the probability that $v$ is matched via a non-crucial edge in $\MM{\mc{G}}$.
	\end{enumerate}
\end{claim}

Consider a non-crucial edge $\{u, v\}$ between two nodes $u$ and $v$ with $d_C(u, v) \geq \lambda$. The probability that $x_e$ is non-zero is $p_e(1-\Pr[X_v])(1-\Pr[X_u])$: Both $u$ and $v$ should be unmatched in $Z$ and $e$ should be realized, and further all these events are independent. This intuitively explains why we set $x_e = \frac{f_e}{p_e(1-\Pr[X_v])(1-\Pr[X_u])}$ if all these conditions hold: We want the denominator to cancel out with this probability so that we get $\E[x_e] = f_e$. We will formalize this intuition in Section~\ref{sec:analysisfractional} where we prove the expected size of $x$ is large.

\subsection{Validity of $x$}\label{sec:validityofx}

In this section, we prove that $x$ is indeed an expected fractional matching of $\mc{Q}$.

First, we prove that $x$ is non-zero only on the edges of $\mc{Q}$. This simply follows from the construction of $x$.

\begin{claim}
	Any edge $e$ with $x_e > 0$ belongs to $\mc{Q}$. That is, $x$ is only non-zero on the set of edges queried by Algorithm~\ref{alg:sampling} that are also realized.
\end{claim}
\begin{proof}
	For any crucial edge $e$, we either have $x_e = 1$ or $x_e = 0$. By definition, if $x_e = 1$ then $e \in Z \cap Q$. By Lemma~\ref{lem:independentmatching}, $Z$ is a matching of {\em realized} crucial edges, i.e. $e \in Z$ implies $e \in \mc{E}$. Therefore, $e \in Z \cap Q$ implies $e \in \mc{E} \cap Q = \mc{Q}$ as desired.
	
	For any non-crucial edge $e$, if $e \not\in Q$, then $f_e = 0$ by definition of $f_e$. Therefore, if $x_e > 0$, then $f_e > 0$ which implies $e \in Q$. Moreover, by (\ref{eq:noncrucialxe}), $x_e > 0$ implies $e$ is realized. Combining these two, we get that if $x_e>0$ then $e \in \mc{Q}$.
\end{proof}

Next, we prove condition (1) of Definition~\ref{def:expfractional}.

\begin{claim}\label{cl:expxvlt1}
	For every vertex $v$, $\E[x_v] \leq 1$.
\end{claim}
\begin{proof}
	Suppose at first that there is an edge $e$ incident to $v$ that belongs to matching $Z$. Then we either have $x_e = 1$ or $x_e = 0$ (depending on whether $e \in Q$ or not). For all other edges $e'$ connected to $v$ (crucial or non-crucial) we have $x_{e'} = 0$ by (\ref{eq:crucialxe}) and (\ref{eq:noncrucialxe}). Therefore if such edge $e$ exists, we indeed have $x_v \leq 1$. For the rest of the proof, we condition on the event that no such edge $e$ exists, i.e. $v \not\in V(Z)$ and prove the claim.
	
	Let $u_1, u_2, \ldots, u_r$ be neighbors of $v$ in graph $G$ such that for all $i \in [r]$: (1) edge $\{v, u_i\}$ is non-crucial, (2) $d_C(v, u_i) \geq \lambda$. Let $e_i := \{v, u_i\}$; we  claim that conditioned on $v \not\in V(Z)$, we have
	\begin{equation}\label{eq:21410238471098412}
	x_v = x_{e_1} + x_{e_2} + \ldots + x_{e_r}.
	\end{equation}
	To see this, fix an edge $e = \{v, u\}$ for some $u \not\in \{u_1, \ldots, u_r\}$. We show that $x_e = 0$, which suffices to prove (\ref{eq:21410238471098412}). First if $e$ is crucial, then $e \not\in Z$ given that $v \not\in V(Z)$; thus according to (\ref{eq:crucialxe}) we set $x_e = 0$. Moreover, if $e$ is non-crucial, the assumption $u \not\in \{u_1, \ldots, u_r\}$ implies $d_C(v, u) < \lambda$ by definition of the set. In this case also, we set $x_e = 0$ according to (\ref{eq:noncrucialxe}); concluding the proof of (\ref{eq:21410238471098412}).
	
	By linearity of expectation applied to (\ref{eq:21410238471098412}), we get
	\begin{equation}\label{eq:612903247}
	\E[x_v \mid v \not\in V(Z)] = \sum_{i=1}^r \E[x_{e_i} \mid v \not\in V(Z)].
	\end{equation}
	Moreover, for any arbitrary $i \in [r]$ we have
	\begin{align}
		\nonumber\E[x_{e_i} \mid v\not\in V(Z)] &= \Pr[u_i \not\in V(Z), e_i \text{ realized} \mid v \not\in V(Z)] \times \frac{\E[f_{e_i}]}{p_{e_i}(1-\Pr[X_v])(1-\Pr[X_{u_i}])}\\
		&= p_{e_i}(1-\Pr[X_{u_i}]) \times \frac{\E[f_{e_i}]}{p_{e_i}(1-\Pr[X_v])(1-\Pr[X_{u_i}])} = \frac{\E[f_{e_i}]}{1-\Pr[X_v]}.\label{eq:7128312098}
	\end{align}
	The second equality above follows from the fact that the event of $e_i$ being realized is independent of $u_i$ or $v$ being in $V(Z)$, as indicated by Lemma~\ref{lem:independentmatching} part 3; and also the fact that $u_i \not\in V(Z)$ and $v \not\in V(Z)$ are also independent from each other due to Lemma~\ref{lem:independentmatching} part 4 combined with the assumption that $d_C(u_i, v) \geq \lambda$. We also note that we have used $\E[f_{e_i}]$ instead of $\E[f_{e_i} \mid v\not\in V(Z)]$ in the equation above since $f_{e_i}$ is only a random variable of the randomization used in Algorithm~\ref{alg:sampling} whereas the matching $Z$ is constructed in Lemma~\ref{lem:independentmatching} independent of the outcome of Algorithm~\ref{alg:sampling}.
	
	Combining (\ref{eq:612903247}) and (\ref{eq:7128312098}) we get
	\begin{equation}\label{eq:421610986201363}
		\E[x_v \mid v\not\in V(Z)] = \sum_{i = 1}^r \frac{\E[f_{e_i}]}{1-\Pr[X_v]} = \frac{1}{1-\Pr[X_v]}\sum_{i=1}^r\E[f_{e_i}].
	\end{equation}	
	From Claim~\ref{cl:frange} part 1, we know $\E[f_{e_i}] \leq q_{e_i}$. Replacing this into the equality above, we get
	$$
		\E[x_v \mid v\not\in V(Z)] \leq \frac{1}{1-\Pr[X_v]} \sum_{i=1}^r q_{e_i} \leq \frac{n_v}{1-\Pr[X_v]}.
	$$
	
	Lemma~\ref{lem:independentmatching} part (2) guarantees that $\Pr[X_v] < c_v$ which implies $1-\Pr[X_v] > 1-c_v$. On the other hand, $c_v + n_v$ is upper bounded by the probability that $v$ is matched in $\opt$, thus $c_v + n_v \leq 1$, implying $n_v \leq 1-c_v$. These, combined with the equation above, gives
	$$
		\E[x_v \mid v\not\in V(Z)] \leq \frac{n_v}{1-\Pr[X_v]} \leq \frac{1-c_v}{1-c_v} = 1.
	$$
	Recalling also that $\E[x_v \mid v\in V(Z)] \leq 1$ as described at the start of the proof, this concludes the proof of the claim that $\E[x_v] \leq 1$.
\end{proof}

Next, we show that condition (2) of Definition~\ref{def:expfractional} also holds for our construction.

\begin{claim}\label{cl:xblossom}
	For all subsets $U \subseteq V$ with $|U|\leq 1/\epsilon$, $x(U) \leq \lfloor \frac{|U|}{2} \rfloor$ with probability 1.
\end{claim}
\begin{proof}
	By definition of $x$, the value of $x_e$ on crucial edges is either 1 or 0. Moreover, the definition also implies that if a vertex $v$ is incident to a crucial edge $e$ with $x_e = 1$, for all other edges $e'$ incident to $v$ we have $x_{e'} = 0$. Call all such vertices {\em integrally matched}. Fix a subset $U$ and let $U'$ be the subset of $U$ excluding its integrally matched vertices. One can easily confirm that if $x(U) > \lfloor |U|/2 \rfloor$, then also $x(U') > \lfloor |U'|/2 \rfloor$. Therefore, either the claim holds, or there should exist a subset with no integrally matched vertices that violates it. Let $U$ be the smallest such subset and observe that $|U| \leq 1/\epsilon$ (otherwise $U$ does not contradict the claim's statement).	
	
	Since $U$ has no integrally matched vertex, for every crucial edge $e$ inside $U$ we have $x_e = 0$ and for every non-crucial edge $e$ inside $U$ by definition (\ref{eq:noncrucialxe}) we have
	$
		x_e \leq \frac{f_e}{p_e (1-\Pr[X_u]) (1-\Pr[X_v])}.
	$
	By definition of $f_e$, it holds that $f_e \leq 1/\sqrt{\epsilon R}$ and by Lemma~\ref{lem:independentmatching} part 2, $\Pr[X_u], \Pr[X_v] \leq 1-\epsilon^2$. Replacing these into the bound above, we get
	$
		x_e \leq \frac{1}{p \times \epsilon^2 \times \epsilon^2 \sqrt{\epsilon R}}.
	$ Noting from Corollary~\ref{cor:thresholds} part 4 that $\tau_- < (\epsilon p)^{50}$ and that $R = 2/\tau_-$, we get $R > 2/(\epsilon p)^{50}$. Replacing this into the previous upper bound on $x_e$, we get that $x_e$ is much smaller than say $\epsilon^3$.
	
	Now since $|U| \leq 1/\epsilon$ there are at most $\binom{|U|}{2} < 1/\epsilon^2$ edges $e$ inside $U$ that can have non-zero $x_e$. For each of these, as discussed above $x_e < \epsilon^3$. Thus we have $x(U) < \epsilon^3 \times 1/\epsilon^2 < 1$ which cannot be larger than $\lfloor |U|/2 \rfloor$ if $|U| \geq 2$ (if $|U| \leq 1$, then there are no edges with both endpoints in $U$ and thus clearly $x(U) = 0$). This contradicts the assumption that $x(U) > \lfloor |U|/2 \rfloor$, implying that there is no such subset. 
\end{proof}

\subsection{The Expected Size of $x$}\label{sec:analysisfractional}

In this section we prove the following.

\begin{lemma}\label{lem:sizeofx}
	It holds that $\E\left[|x| \right] \geq (1-34\epsilon)\opt$.
\end{lemma}

We start by analyzing the size of $x$ on the crucial edges. This is a simple consequence of Lemma~\ref{lem:independentmatching} part 1 which guarantees $\E[Z]\geq q(C)-30\epsilon \opt$ and Observation~\ref{obs:crucialdegree} which guarantees each crucial edge belongs to $Q$ with probability at least $1-\epsilon$.

\begin{claim}\label{cl:sizeofxcrucial}
	It holds that $\E\left[\sum_{e \in C} x_e \right] \geq q(C)-31\epsilon \opt$.
\end{claim}
\begin{proof}
	Denoting $x(C) = \sum_{e \in C} x_e$, we have
	\begin{equation*}
		\E[x(C)] = \E \Big[ \sum_{e \in C} x_e \Big] = \sum_{e \in C} \E[x_e] = \sum_{e \in C} \Pr[e \in Q \text{ and } e \in Z].
	\end{equation*}
	Observe that $Z$ and $Q$ are picked independently as Lemma~\ref{lem:independentmatching} is essentially unaware of $Q$. Therefore, for any crucial edge $e$ we get 
	$$
	\Pr[e \in Q \text{ and } e \in Z] = \Pr[e \in Q] \times \Pr[e \in Z] \geq (1-\epsilon)\Pr[e \in Z],
	$$
	where the latter inequality comes from Observation~\ref{obs:samplealmostallcrucial}. Replacing this to the equality above gives
	$$
	\E[x(C)] \geq (1-\epsilon)\sum_{e \in C} \Pr[e \in Z] = (1-\epsilon)\E[|Z|] \stackrel{\text{Lemma~\ref{lem:independentmatching} part 1}}{\geq} (1-\epsilon) (q(C)-30\epsilon \opt) \geq q(C) - 31\epsilon \opt,
	$$
	completing the proof of the claim.
\end{proof}

To analyze the size of $x$ on the non-crucial edges, we first define $N'$ to be the subset of non-crucial edges $\{u, v\}$ such that $d_C(u, v) \geq \lambda$ and define $q(N') := \sum_{e \in N'} q_e$ and $x(N') := \sum_{e \in N'} x(N')$. Definition of $N'$ is useful since recall from (\ref{eq:noncrucialxe}) that for any $\{u, v\} \in N$ with $d_C(u, v) < \lambda$ (i.e. $\{u, v\} \not\in N'$) we set $x_e = 0$. Therefore only the edges in $N$ that also belong to $N'$ have non-zero $x_e$, implying $x(N) = x(N')$.

\begin{claim}\label{cl:nplarge}
	It holds that $q(N') \geq q(N)-\epsilon q(C)$.
\end{claim}
\begin{proof}
	For any edge $e = \{u, v\}$ in $N \setminus N'$, we choose an arbitrary shortest path $P$ between $u$ and $v$ in graph $C$ and charge the edges of this path. Note that by definition of $N'$, such path between $u$ and $v$ exists and has size less than $\lambda$. Now, take a crucial edge $f$. We denote by $\Phi(f)$ the set of edges in $N \setminus N'$ for which we charge a path containing $f$. Below, we argue that
	\begin{equation}\label{eq:612398273497}
		|\Phi(f)| \leq 4(1/\tau_+)^{2\lambda} \qquad \forall f \in C.
	\end{equation}
	
	Fix a crucial edge $f$ and an edge $\{u, v\} \in \Phi(f)$. As discussed above, there should be a path of length less than $\lambda$ between $u$ and $v$ in graph $C$ that passes through $f$. This means that $d_C(u, f) < \lambda$ and $d_C(v, f) < \lambda$. Therefore, both $u$ and $v$ are at distance at most $\lambda$ from $f$ in graph $C$. 
	
	Observe that there are at most $2(\Delta_C)^{\lambda}$ vertices in the $\lambda$-neighborhood of $f$ in graph $C$. Thus, there are at most $2(\Delta_C)^{\lambda} \times 2(\Delta_C)^{\lambda} = 4(\Delta_C)^{2\lambda}$ pairs of vertices that can potentially charge $f$, proving $|\Phi(f)| \leq 4(\Delta_C)^{2\lambda} \leq 4(1/\tau_+)^{2\lambda}$ where the latter inequality comes from Observation~\ref{obs:crucialdegree} that $\Delta_C \leq 1/\tau_+$. This concludes the proof of (\ref{eq:612398273497}).
	
	As discussed above, each edge $e \in N \setminus N'$ charges a path in $C$, thus belongs to $\Phi(f)$ of at least one crucial edge $f$. Therefore, we get
	\begin{equation}\label{eq:10234817293478}
		|N \setminus N'| \leq \sum_{f \in C} \Phi(f).
	\end{equation}
	Every edge $e$ in $N \setminus N'$ is non-crucial, i.e. $q_e \leq \tau_-$. Thus:
	\begin{equation}\label{eq:7873241712304912348}
	\sum_{e \in N \setminus N'} q_e \leq \tau_-|N \setminus N'| \stackrel{(\ref{eq:10234817293478})}{\leq} \tau_- \sum_{f \in C} \Phi(f) \stackrel{(\ref{eq:612398273497})}{\leq} 4\tau_- |C|(1/\tau_+)^{2\lambda} \leq 4\tau_- q(C)(1/\tau_+)^{2\lambda+1},
	\end{equation}
	where the last inequality comes from the fact that $q(C) \geq |C| \tau_+$ as for every edge $e \in C$, $q_e \geq \tau_+$.
	
	From Corollary~\ref{cor:thresholds} we have $\tau_- = (\tau_+)^{10g}$ and we have $g \geq \lambda > 1$ by Observation~\ref{obs:ggtlambda}. Thus:
	$$
	 4\tau_- (1/\tau_+)^{2\lambda+1} = 4 (\tau_+)^{10g} (1/\tau_+)^{2\lambda+1} = 4 (\tau_+)^{10g - (2\lambda - 1)} < 4 \tau_+ < \epsilon.
	$$
	Replacing it into inequality (\ref{eq:7873241712304912348}), we get
	$$
	\sum_{e \in N \setminus N'} q_e \leq \epsilon q(C).
	$$
	This concludes the proof since
	$$
		q(N') = \sum_{e \in N'} q_e = \sum_{e \in N \setminus (N \setminus N')} q_e \geq \sum_{e \in N}q_e - \sum_{e \in N \setminus N'} q_e \geq q(N) - \epsilon q(C)
	$$
	as it is desired.
\end{proof}

\begin{claim}\label{cl:xnpgtqnp}
	It holds that $\E[x(N')] \geq (1-\epsilon) q(N')$.
\end{claim}
\begin{proof}
	By linearity of expectation, we have 
	\begin{equation}\label{eq:16234421340}
	\E[x(N')] = \E \Big[ \sum_{e \in N'} x_e \Big] = \sum_{e \in N'} \E[x_e].
	\end{equation}
	We emphasize that the expectation here is taken over the randomization in Algorithm~\ref{alg:sampling}, the randomization in matching $Z$, and the randomization in realization of non-crucial edges. Specifically, we write $\E_{\alg, Z, \mc{N}}[x_e]$ to emphasize on this.
	
	The randomization of Algorithm~\ref{alg:sampling} determines the value of $f_e$ which is used in defining $x_e$. Let us first condition on $f_e$ and compute $\E_{Z, \mc{N}}[x_e \mid f_e]$. We have
	\begin{equation}\label{eq:89123}
		\E_{Z, \mc{N}}[x_e \mid f_e] = \Pr[e \in \mc{E} \text{ and } u, v \not\in V(Z) \mid f_e] \times  \frac{f_e}{p_e(1-\Pr[X_u])(1-\Pr[X_v])}.
	\end{equation}
	We claim that 
	\begin{equation}\label{eq:5123674182374}
		\Pr[e \in \mc{E} \text{ and } u, v \not\in V(Z) \mid f_e] = p_e(1-\Pr[X_u])(1-\Pr[X_v]).
	\end{equation}
	To see this, first observe that the value of $f_e$ is determined solely by the random realizations taken by Algorithm~\ref{alg:sampling}. In particular, the events $e \in \mc{E}$, and $u, v \not\in V(Z)$ are completely independent of the outcome of Algorithm~\ref{alg:sampling}. This allows us to remove the condition on $f_e$ from the left hand side of (\ref{eq:5123674182374}). Moreover, by Lemma~\ref{lem:independentmatching} part 3, the matching $Z$ is chosen independently from the realization of non-crucial edges, thus events $e \in \mc{E}$ and $u, v \not\in V(Z)$ are independent. 	Finally, the assumption that $e \in N'$, by definition of $N'$, implies that $d_C(u, v) \geq \lambda$. Therefore, by Lemma~\ref{lem:independentmatching} part 4, events $v \in V(Z)$ and $u \in V(Z)$ (and for that matter their complements) are independent. Thus, indeed:
	\begin{align*}
	\Pr[e \in \mc{E} \text{ and } u, v \not\in V(Z) \mid f_e] &= \Pr[e \in \mc{E}] \times \Pr[v \not\in V(Z)] \times \Pr[u \not\in V(Z)]\\
	&= p_e(1-\Pr[X_u])(1-\Pr[X_v]).
	\end{align*}
	Replacing (\ref{eq:5123674182374}) into (\ref{eq:89123}) we get
	\begin{equation*}
		\E_{Z, \mc{N}}[x_e \mid f_e] = p_e(1-\Pr[X_u])(1-\Pr[X_v]) \times \frac{f_e}{p_e(1-\Pr[X_u])(1-\Pr[X_v])} = f_e.
	\end{equation*}
	Taking expectation over $\alg$ from both sides, we get
	\begin{equation}\label{eq:72313409}
	\E_{\alg}[\E_{Z, \mc{N}}[x_e \mid f_e]] = \E_{\alg}[f_e].
	\end{equation}
	The left hand side equals $\E_{\alg, Z, \mc{N}}[x_e]$. For the right hand side, by Claim~\ref{cl:frange} we have $\E[f_e] \geq (1-\epsilon)q_e$. 	Replacing both the left hand side and right hand side of (\ref{eq:72313409}) by these bounds, we get
	\begin{equation}
		\E_{\alg, Z, \mc{N}}[x_e] \geq (1-\epsilon) q_e.
	\end{equation}
	Combining this with (\ref{eq:16234421340}) we get
	\begin{equation*}
	\E[x(N')] = \sum_{e \in N'} \E[x_e] \geq (1-\epsilon) \sum_{e \in N'} q_e = (1-\epsilon) q(N'),
	\end{equation*}
	completing the proof.
\end{proof}

We are now ready to prove Lemma~\ref{lem:sizeofx}.

\begin{proof}[Proof of Lemma~\ref{lem:sizeofx}]
	We have
	$$
	\E\Big[\sum_{e}x_e\Big] = \E\Big[\sum_{e \in C} x_e\Big] + \E\Big[\sum_{e \in N} x_e\Big] \stackrel{\text{Claim~\ref{cl:sizeofxcrucial}}}{\geq} q(C) - 31\epsilon \opt + \E\Big[\sum_{e \in N} x_e\Big].
	$$
	Also note that for $e \in N$, $x_e \not= 0$ iff $e \in N'$ by construction of $x$. Thus,
	$$
	\E\Big[\sum_{e \in N} x_e\Big] = 	\E\Big[\sum_{e \in N'} x_e\Big] = \E[x(N')] \stackrel{\text{Claim~\ref{cl:xnpgtqnp}}}{\geq} (1-\epsilon)q(N') \stackrel{\text{Claim~\ref{cl:nplarge}}}{\geq} (1-\epsilon)(q(N)-\epsilon q(C)).
	$$
	Combining the two equations above, we get
	\begin{align*}
		\E\Big[\sum_{e}x_e\Big] &\geq q(C) - 31\epsilon \opt + (1-\epsilon)(q(N)-\epsilon q(C)) > q(C) + q(N) - 33\epsilon \opt\\
		&\stackrel{\text{Lemma~\ref{lem:gap} part (2)}}{\geq} (1-\epsilon)\opt - 33\epsilon \opt \geq (1-34\epsilon)\opt,
	\end{align*}
	concluding the proof.
\end{proof}

\subsection{From the Expected Fractional Matching to an Actual Fractional Matching}\label{sec:turntofracmatching}

We showed that $x$ is an expected fractional matching satisfying $\E[x_v] \leq 1$ for every vertex $v$. However, as mentioned before, there is still a possibility that $x_v > 1$ depending on the coin tosses of the algorithms and the realization. This should never occur in a valid fractional matching. Thus, we define the following scaled fractional matching $y$ based on $x$ which decreases the fractional matching around vertices that deviate significantly from their expectation to 0.

\begin{equation}\label{eq:defy}
	\text{For any edge $e=\{u, v\}$,} \qquad\qquad y_e = \begin{cases}
			x_e/(1+\epsilon) & \text{if $x_v, x_u \leq 1+\epsilon$,}\\
			0 & \text{otherwise.}
	\end{cases}
\end{equation}

\begin{observation}\label{obs:yblossom}
	By definition above, $y$ is a valid fractional matching, i.e. $y_v \leq 1$ for all $v \in V$. In addition, since $y_e \leq x_e$ for all edges $e$, Claim~\ref{cl:xblossom} implies that for all $U \subseteq V$ with $|U| \leq 1/\epsilon$, $y(X) \leq \lfloor \frac{|U|}{2} \rfloor$. That is, $y$ also satisfies all blossom inequalities of size up to $1/\epsilon$.
\end{observation}

It remains to prove that while turning the expected fractional matching $x$ into an actual fractional matching $y$, we don't significantly hurt the matching's size. We address this in the lemma below.

\begin{lemma}\label{lem:ylarge}
	$\E[|y|] \geq (1-55\epsilon)\opt$.
\end{lemma}

The main ingredient in proving Lemma~\ref{lem:ylarge} is the following claim.

\begin{claim}\label{cl:6123719801923}
	For every vertex $v$, $\Pr[x_v > 1+\epsilon] \leq \epsilon^6p$.
\end{claim}

Let us first see how Claim~\ref{cl:6123719801923} suffices to prove Lemma~\ref{lem:ylarge} and then prove it.

\begin{proof}[Proof of Lemma~\ref{lem:ylarge}]
	We have
	\begin{flalign*}
		\sum_{e} y_e &= \sum_{e = \{u, v\}} \mathbbm{1}(x_u \leq 1+\epsilon \text{ and } x_v \leq 1+\epsilon) \frac{x_e}{1+\epsilon} && \text{By definition of $y_e$ in (\ref{eq:defy}).}\\
		&\geq \sum_{e = \{u, v\}} (1-\mathbbm{1}(x_u > 1+\epsilon)-\mathbbm{1}(x_v > 1+\epsilon)) \frac{x_e}{1+\epsilon} && \text{Union bound.}\\
		&= \sum_{e} \frac{x_e}{1+\epsilon} - 2\sum_{v : x_v > 1+\epsilon} \sum_{e \ni v} \frac{x_e}{1+\epsilon} = \sum_{e} \frac{x_e}{1+\epsilon} - 2\sum_{v : x_v > 1+\epsilon} \frac{x_v}{1+\epsilon}.
	\end{flalign*}
	Taking expectation from both sides, we get
	\begin{flalign}
		\nonumber\E\Big[ \sum_e y_e \Big] &\geq \E\Big[\sum_{e} \frac{x_e}{1+\epsilon} - 2\sum_{v : x_v > 1+\epsilon} \frac{x_v}{1+\epsilon}\Big] = \frac{1}{1+\epsilon}\left(\E\Big[\sum_{e} x_e\Big] - 2\E\Big[\sum_{v : x_v > 1+\epsilon} x_v \Big]\right)\\
		\nonumber&\geq \frac{1}{1+\epsilon}\left((1-34\epsilon)\opt - 2\E\Big[\sum_{v : x_v > 1+\epsilon} x_v \Big]\right) \qquad\qquad \text{By Lemma~\ref{lem:sizeofx}.}\\
		\nonumber &\geq (1-35\epsilon)\opt - 2\sum_{v} \Pr[x_v > 1+\epsilon]\E[x_v \mid x_v > 1+\epsilon]\\
		&\geq (1-35\epsilon)\opt - 2\sum_{v} \epsilon^6 p \E[x_v \mid x_v > 1+\epsilon] \qquad\qquad \text{By Claim~\ref{cl:6123719801923}.}\label{eq:98912308}
	\end{flalign}
	We will soon prove that for every vertex $v$, it \underline{deterministically} holds that $x_v \leq \frac{1}{p\epsilon^4}$. Replacing this into the last inequality above, gives the desired bound that
	\begin{flalign*}
		\E\Big[\sum_e y_e \Big] &\geq (1-35\epsilon)\opt - 2\sum_{v} \epsilon^6 p \frac{1}{p \epsilon^4} \geq (1-35\epsilon)\opt - 2\epsilon^2 n \stackrel{\text{Assumption~\ref{ass:optlarge}}}{\geq} (1-35\epsilon)\opt - 20\epsilon \opt \\
		&= (1-55\epsilon)\opt.
	\end{flalign*}
	Now let's see why $x_v \leq \frac{1}{p\epsilon^4}$. Observe from the definition of $x$ that if $v \in V(Z)$ then $x_v \leq 1$ and otherwise
	$$
		x_v = \sum_{e = \{v, u\}} x_e \leq \sum_{e = \{v, u\}} \frac{f_e}{p(1-\Pr[X_u])(1-\Pr[X_v])} \leq \frac{1}{p \epsilon^4} \sum_{e=\{v, u\}} f_e.
	$$
	The last inequality above comes from the fact that for every vertex $w$, $\Pr[X_w] \leq 1-\epsilon^2$ due to Lemma~\ref{lem:independentmatching} part 2, which means $1-\Pr[X_w] \geq \epsilon^2$. 
	
	Now recall from Claim~\ref{cl:frange} part 3 that $\sum_{e \ni v} f_e \leq 1$. Thus we get our desired upper bound that $x_v \leq \frac{1}{p\epsilon^4}$.	 As described above, this completes the proof that $\E[\sum_e y_e] \geq (1-55\epsilon)\opt$.
\end{proof}

We now turn to prove Claim~\ref{cl:6123719801923} that $\Pr[x_v > 1+\epsilon] \leq \epsilon^6 p$ for all $v$.

\newcommand{\eventvz}[0]{\ensuremath{A}}

\begin{proof}[Proof of Claim~\ref{cl:6123719801923}]
	If an edge incident to $v$ belongs to matching $Z$, i.e. if $X_v = 1$ (as defined in Lemma~\ref{lem:independentmatching}), then one can confirm easily from the definition of $x$ in (\ref{eq:crucialxe}) and (\ref{eq:noncrucialxe}) that either $x_v = 1$ or $x_v = 0$, implying that $\Pr[x_v > 1+\epsilon \mid X_v = 1] = 0$. As such, for the rest of the proof, we simply condition on the event that $X_v = 0$.
	
	Similar to the proof of Claim~\ref{cl:expxvlt1} let $u_1, u_2, \ldots, u_r$ be the neighbors of $v$ such that for each $i \in [r]$, (1) edge $e_i = \{v, u_i\}$ is non-crucial, and (2) $d_C(v, u_i) \geq \lambda$. Recall from (\ref{eq:21410238471098412}) that given event $X_v = 0$, it holds that
		\begin{equation*}
			x_v = x_{e_1} + x_{e_2} + \ldots + x_{e_r}.
		\end{equation*}
	
		Let $f'_v := \sum_{i=1}^r f_{e_i}$ and note that $f'_v \leq f_v$ since $f_v$ is sum of $f_e$ of all non-crucial edges $e$ connected to $v$. Claim~\ref{cl:frange} part 4 proves that $\Pr[f_v \geq n_v + 0.1\epsilon] \leq (\epsilon p)^{10}$. Therefore, it also holds that $\Pr[f'_v \geq n_v + 0.1\epsilon] \leq (\epsilon p)^{10}$ since $f'_v \leq f_v$. For the rest of the proof, we regard $f_{e_i}$'s as (adversarially) fixed with the only assumption that $f'_v < n_v + 0.1\epsilon$ which happens with probability at least $1 - (\epsilon p)^{10}$. We denote this event, as well as the event that $X_v = 0$, by $\eventvz$ and prove
		\begin{equation}\label{eq:980989825627}
		\Pr[x_v > 1 + \epsilon \mid \eventvz] \leq 0.5\epsilon^6p,
		\end{equation}
		which clearly is sufficient for proving the claim.
	
		We do this by proving a concentration bound using the second moment method. Consider the variance of $x_v$ conditioned on $\eventvz$:
		\begin{flalign*}
			\Var[x_v \mid \eventvz] = \sum_{i=1}^r \sum_{j=1}^r \Cov(x_{e_i}, x_{e_j} \mid \eventvz).
		\end{flalign*}
		Now that $f_e$'s are fixed, $x_v$ is only a random variable of (1) the randomization used in Lemma~\ref{lem:independentmatching} for obtaining matching $Z$, and (2) the realization of non-crucial edges.
		
		In what follows we identify a condition under which covariance of $x_{e_i}$ and $x_{e_j}$ becomes $0$. We will use this later to upper bound $\Var[x_v \mid \eventvz]$. 
		
		\begin{observation}\label{obs:cov0}
			Let $i, j \in [r]$ be such that $d_C(u_i, u_j) \geq \lambda$. Then $\Cov(x_{e_i}, x_{e_j} \mid \eventvz) = 0$.
		\end{observation}
		\begin{proof}
			We already had $d_C(v, u_i) \geq \lambda$ and $d_C(v, u_j) \geq \lambda$ by definition of $u_i, u_j$. Combined with assumption $d_C(u_i, u_j) \geq \lambda$ and using Lemma~\ref{lem:independentmatching} part 4, we get that $X_v, X_{u_i}, X_{u_j}$ are independent. Realization of $e_i$ and $e_j$ are also independent even given $\eventvz$. This is because these are non-crucial edges and thus are realized independently from $Z$ (according to Lemma~\ref{lem:independentmatching} part 3) or the values of $f$ which are derived from Algorithm~\ref{alg:sampling}.
			
			By definition (\ref{eq:noncrucialxe}), the value of $x_{e_i}$ conditioned on $\eventvz$ is fully determined once we know $X_{u_i}$ and whether $e_i$ is realized. Similarly, the value of $x_{e_j}$ conditioned on $\eventvz$ is fully determined once we know $X_{u_j}$ and whether $e_j$ is realized. These, as discussed above, are independent. Hence $x_{e_i}$ and $x_{e_j}$, conditioned on $\eventvz$, are independent and thus their covariance is 0.
		\end{proof}
		
		Now consider two vertices $u_i$ and $u_j$ (possibly $u_i = u_j$) where $d_C(u_i, u_j) < \lambda$. Here, the covariance may not be 0. But we still can upper bound it as follows:
		\begin{flalign}
		\nonumber \Cov(x_{e_i}x_{e_j} \mid \eventvz) &= \E[x_{e_i}x_{e_j} \mid \eventvz] - \E[x_{e_i} \mid A]\E[x_{e_j} \mid A] \leq \E[x_{e_i}x_{e_j} \mid \eventvz]\\
		\nonumber &\leq \frac{f_{e_i}}{p(1-\Pr[X_v])(1-\Pr[X_{u_i}])} \times \frac{f_{e_j}}{p(1-\Pr[X_v])(1-\Pr[X_{u_j}])}\\
		&\leq \frac{f_{e_i}f_{e_j}}{p^2 \epsilon^8},\label{eq:3819123987}
		\end{flalign}
		where the last inequality follows from Lemma~\ref{lem:independentmatching} part 2 that states for all vertices $w$, $\Pr[X_w] < 1-\epsilon^2$ and thus $1-\Pr[X_w] \geq \epsilon^2$.
		
		Now, for each $i \in [r]$, let $D_i := \{j : d_C(u_i, u_j) < \lambda \}$. Since $C$ is a graph of max degree $\Delta_C$, the $\lambda-1$ neighborhood of each vertex $u_i$ in $C$ includes $\leq (\Delta_C)^{\lambda-1}$ vertices. Thus:
		\begin{equation}\label{eq:87193107123897}
			|D_i| \leq (\Delta_C)^{\lambda-1} \qquad\qquad \text{for every $i \in [r]$.}
		\end{equation}
		Having these, we obtain that
		\begin{flalign*}
			\Var[x_v \mid \eventvz] &= \sum_{i=1}^r \sum_{i=1}^r \Cov(x_{e_i}, x_{e_j} \mid \eventvz) \stackrel{\text{Obs~\ref{obs:cov0}}}{=} \sum_{i = 1}^r \sum_{j \in D_i} \Cov(x_{e_i}, x_{e_j} \mid \eventvz) \stackrel{(\ref{eq:3819123987})}{\leq} \sum_{i = 1}^r \sum_{j \in D_i} \frac{f_{e_i}f_{e_j}}{p^2\epsilon^8} \\
			& = \frac{1}{p^2 \epsilon^8}\sum_{i = 1}^r \Big(f_{e_i}\sum_{j \in D_i} f_{e_j}	\Big) \stackrel{f_{e_j} \leq \frac{1}{\sqrt{\epsilon R}} \text{ by (\ref{eq:deff})}}{\leq} \frac{1}{p^2\epsilon^8} \sum_{i = 1}^r \Big(f_{e_i} |D_i| \frac{1}{\sqrt{\epsilon R}}	\Big)\\
			& \stackrel{(\ref{eq:87193107123897})}{\leq} \frac{(\Delta_C)^{\lambda-1}}{p^2 \epsilon^8\sqrt{\epsilon R}} \sum_{i = 1}^r f_{e_i} \stackrel{\text{Claim~\ref{cl:frange} part 3}}{\leq} \frac{(\Delta_C)^{\lambda-1}}{p^2 \epsilon^8\sqrt{\epsilon R}} \stackrel{\text{Obs~\ref{obs:crucialdegree}}}{\leq} \frac{(1/\tau_+)^{\lambda-1}}{p^2\epsilon^{8.5}\sqrt{R}}.
		\end{flalign*}
		Replacing $R$ with $\frac{1}{2\tau_-}$ and noting that $\tau_- = (1/\tau_+)^{10 g}$, we get that
		\begin{flalign*}
		\Var[x_v \mid A] \leq \frac{2(1/\tau_+)^{\lambda}}{p^2 \epsilon^{8.5}(1/\tau_+)^{10g}} &= \frac{2}{p^2 \epsilon^{8.5}}(\tau_+)^{10g - \lambda}\\
		&< \frac{2\tau_+}{p^2 \epsilon^{8.5}} && \text{By Observation~\ref{obs:ggtlambda} $g \geq \lambda > 1$ and $\tau_+ < 1$.}\\
		&< \frac{2 (\epsilon p)^{50}}{p^2 \epsilon^{8.5}} && \text{Corrolary~\ref{cor:thresholds} part 4.}\\
		&= 2 \epsilon^{41.5} p^{48} < 0.1 \epsilon^8 p.
		\end{flalign*}
		With this upper bound on the variance, we can use Chebyshev's inequality to get
		\begin{equation}\label{eq:cheb18023}
			\Pr\Big[|x_v - \E[x_v \mid \eventvz]| > 0.5\epsilon \,\Big\vert\, \eventvz\Big] \leq \frac{\Var[x_v \mid \eventvz]}{(0.5\epsilon)^2} \leq \frac{0.1 \epsilon^8 p}{0.25 \epsilon^2} < 0.5\epsilon^6 p.
		\end{equation}
		Next, recall from (\ref{eq:421610986201363}) in the proof of Claim~\ref{cl:expxvlt1} that $\E[x_v \mid v\not\in V(Z)] \leq \frac{\sum_{i=1}^r \E[f_{e_i}]}{1-\Pr[X_v]} = \frac{f'_v}{1-\Pr[X_v]}$. Event $\eventvz$ in addition to $v \not\in V(Z)$ also fixes the value of $f'_v$. But recall that event $\eventvz$ (as we defined it) guarantees $f'_v \leq n_v + 0.5\epsilon$. Therefore, we get
		\begin{equation}\label{eq:74777748123}
			\E[x_v \mid \eventvz] \leq \frac{n_v + 0.5\epsilon}{1-\Pr[X_v]} \stackrel{\Pr[X_v] < c_v}{\leq} \frac{n_v + 0.5\epsilon}{1-c_v} \stackrel{n_v \leq 1-c_v}{\leq} \frac{1-c_v+0.5\epsilon}{1-c_v} \leq 1 + 0.5\epsilon.
		\end{equation}
		Combining (\ref{eq:cheb18023}) and (\ref{eq:74777748123}) we get the claimed inequality of (\ref{eq:980989825627}) that 
		$$
		\Pr[x_v > 1+\epsilon \mid \eventvz] \leq \Pr[|x_v - \E[x_v \mid A]| > 0.5\epsilon \mid \eventvz] \leq 0.5\epsilon^6p,
		$$
		which as described before suffices to prove $\Pr[x_v > 1+\epsilon] \leq \epsilon^6p$.
\end{proof}

\section{Proof of the Vertex-Independent Matching Lemma}\label{sec:independentmatching}

\newcommand{\neighbors}[1]{\ensuremath{\mathsf{Neighbors}(#1)}}
\newcommand{\isrealized}[1]{\ensuremath{\mathsf{IsRealized}(#1)}}
\newcommand{\dependent}[0]{\ensuremath{\mathcal{D}}}

In this section we turn to prove Lemma~\ref{lem:independentmatching} restated below.

\restate{Lemma~\ref{lem:independentmatching}}{
	\independentmatching{}
}

\subsection{Overview of the Algorithm}\label{sec:crucialoverview}

In this section, we give an overview of our algorithm for proving Lemma~\ref{lem:independentmatching}. We emphasize that the overview given here is deliberately informal to describe the main intuitions, with the hope that it makes the algorithm and its analysis more accessible.

Satisfying property~3 required by Lemma~\ref{lem:independentmatching} turns out to be easy. Recall that we are constructing matching $Z$ on the realized {\em crucial} edges, thus we can simply ignore realization of non-crucial edges and automatically satisfy property~3. Among the other 3, let us first focus on property~4. How can we argue that the output matching satisfies the required independence property? We show that the \local{} model of computation can be naturally used for this purpose. We start with the formal definition of the model and then describe how it can be used in this case.

\smparagraph{The \local{} model \cite{DBLP:journals/siamcomp/Linial92}.} In the \local{} model, the input is a graph and there is a processor on each node of this graph. Computation proceeds in synchronous rounds and in each round, each processor can send a message (of any size) to each of its neighbors. The goal is to output a property of this communication graph, e.g. a matching of it. At the end, each node should know its part of the output, e.g. which one of its edges, if any, is part of the matching.

\smparagraph{Why the \local{} model.} A particularly useful property of any $r$-round \local{} algorithm is that the output of each node essentially depends only on its $r$-hop neighborhood. That is, having the $r$-hop neighborhood of each node $v$ (including the random tapes of the nodes in the neighborhood), we can uniquely determine the output of $v$. Therefore if the shortest path between two nodes is at least $2r+1$, their outputs are essentially independent of each other after $r$ rounds. 

This is how we prove property~4 of Lemma~\ref{lem:independentmatching} is satisfied: We give a \local{} algorithm operating on graph $C$ where each vertex is initially only aware of the realization of its incident edges. We show that the algorithm within $< \lambda/2$ rounds, finds a matching satisfying the other 3 properties. Then property~4 will be automatically satisfied. That is, for every subset $I$ of the vertices with pairwise distance at least $\lambda$, their outputs will be independent.

\smparagraph{Overview of the algorithm.} The challenge is to ensure that the algorithm has low round-complexity while also satisfying properties 1 and 2. That is, the reported matching $Z$ should be large in expectation (property 1), and that no vertex $v$ should be matched with a larger probability than that specified in property~2. If one ignores the 2nd property, then simply finding a $(1-\epsilon)$-approximate maximum matching in graph $\mc{C}$ will satisfy the first property. And we remark that $O(\log \Delta_C)$-round algorithms (with no dependence on $n$) do exist for this purpose. However, bounding at the same time, the probability that each vertex is matched complicates things. 

Our general idea for the algorithm is as follows: We define a recursive algorithm $\findmatching{r}{\mc{C}}$ (Algorithm~\ref{alg:crucial}) which uses $\findmatching{r-1}{\mc{C}}$ as a subroutine. The base algorithm $\findmatching{0}{\mc{C}}$ returns an empty matching. Let us use $Z_r$ to denote the matching returned by $\findmatching{r}{\mc{C}}$. It will hold that
$$
0 = \E[|Z_0|] \leq \E[|Z_1|] \leq \E[|Z_2|] \leq \ldots 
$$
until eventually for large enough $t = O_\epsilon(1)$, $\E[|Z_t|]$ is desirably large, satisfying property~1. At the same time, we will ensure that for any vertex $v$, the probability that it gets matched in $Z_r$ never exceeds the upper bound of property~2 for any $r$. 

Suppose that for a vertex $v$, we hit this upper bound on the probability that it is matched for algorithm $\findmatching{r}{\mc{C}}$. At this point, we will mark $v$ as {\em saturated} and ensure that we never increase the probability of it being matched. But to keep increasing the matching's size, it may be necessary to say remove a matching edge $\{v_1, v_2\}$ between two saturated vertices $v_1$ and $v_2$, so that we can add two edges $\{v_1, v_3\}$ and $\{v_2, v_4\}$ to the matching where $v_3$ and $v_4$ are unsaturated. Such structures are similar to augmenting paths. However, since the graph is stochastic, these edges $\{v_1, v_2\}, \{v_1, v_3\}, \{v_2, v_4\}$ may not necessarily be part of one realization. We call these natural generalizations of augmenting paths, ``augmenting hyperwalks'' (see Section~\ref{sec:formalcrucialalg}) and show that they can be used to increase the matching size while not increasing probability of saturated vertices getting matched.

In Section~\ref{sec:formalcrucialalg} we present a centralized view of the algorithm. In Section~\ref{sec:p1} we analyze the expected size of the matching returned by this algorithm and argue that it satisfies property~1 of Lemma~\ref{lem:independentmatching}. In Section~\ref{sec:p2} we prove the upper bound on the probability of each vertex getting matched, thereby proving property~2 of Lemma~\ref{lem:independentmatching}. Finally, in Section~\ref{sec:p4} we show that the algorithm has an efficient \local{} implementation, satisfying property~4 of Lemma~\ref{lem:independentmatching}.

\subsection{The Formal Algorithm}\label{sec:formalcrucialalg}

We say $P = ((\mc{C}_0, M_0), \ldots, (\mc{C}_\alpha, M_\alpha))$ is a {\em profile} if each $\mc{C}_i$ is a subgraph of $C$ and each $M_i$ is a matching of $\mc{C}_i$. Furthermore, we call a sequence $W = ((e_1, s_1), \ldots, (e_k, s_k))$ a {\em hyperwalk} of size $k$ if the following conditions hold:
\begin{enumerate}[itemsep=0pt,topsep=5pt]
	\item Each $s_i$ is an integer in $\{0, \ldots, \alpha\}$.
	\item Each $e_i$ is an edge in graph $C$ and sequence $(e_1, e_2, \ldots, e_k)$ is a walk in graph $C$.
\end{enumerate}
We say $P\Delta W := ((\mc{C}_0, M'_0), \ldots, (\mc{C}_\alpha, M'_\alpha))$ is the result of {\em applying} $W$ on $P$ if:
$$
	M'_i = M_i \cup \{ e_j \mid j \text{ is odd, and } s_j = i\} \setminus \{ e_j \mid j \text{ is even, and } s_j = i\}, \qquad \text{for all $i \in \{0, \ldots, \alpha\}$}.
$$
\begin{definition}[Augmenting hyperwalks] \label{def:aug}
For every vertex $v$, let $d_P(v) := \big |\{ i \mid v \in V(M_i) \}\big |$. We say $W$ is an {\em augmenting-hyperwalk} of $P$ if it satisfies the three following conditions.
\begin{enumerate}[itemsep=0pt,topsep=5pt]
	\item $P \Delta W$ is a profile, i.e. each  $M'_i$ in $P \Delta W$ is a matching of graph $\mc{C}_i$.
	\item For all vertices $v$ in walk $(e_1, \ldots, e_k)$ except its first and last vertex, $d_P(v) = d_{P\Delta W}(v)$. 
	\item For the first and last vertices $v$ in walk $(e_1, \ldots, e_k)$, $d_P(v) + 1 = d_{P \Delta W}(v)$.
\end{enumerate}
\end{definition}
Having defined augmenting-hyperwalks, we can now formally state the algorithm---see Algorithm~\ref{alg:crucial}. The algorithm is recursive. Given a realization $\mc{C}$ of $C$, algorithm $\findmatching{r}{\mc{C}}$ uses algorithm $\findmatching{r-1}{\mc{C}}$ as a subroutine and then returns a matching of $\mc{C}$. The base algorithm $\findmatching{0}{\mc{C}}$ returns an empty matching. We will show that for $t = 1/\epsilon^9$, algorithm $\findmatching{t}{\mc{C}}$ satisfies the properties of Lemma~\ref{lem:independentmatching}.

\begin{tboxalgh}{\findmatching{r}{\mc{C}}}\label{alg:crucial}
	\vspace{-0.3cm}
	\begin{enumerate}[label=(\arabic*), itemsep=-2pt, leftmargin=23pt]
		\item If $r = 0$, return $\emptyset$.
		\item Draw $\alpha := 1/\epsilon^7-1$ realizations $\mc{C}_1, \ldots, \mc{C}_\alpha$ of $C$  where each realization $\mc{C}_{i}$ includes each edge $e$ of $C$ independently with probability $p_e$. Also let $\mc{C}_0 := \mc{C}$.
		\item Consider profile $P = ((\mc{C}_0, M_0), \ldots, (\mc{C}_\alpha, M_\alpha))$ where $M_i = \findmatching{r-1}{\mc{C}_i}$.
		\item For every vertex $v$, define $\gamma_{v, r-1} := \Pr[\text{$v$ is matched in \findmatching{r-1}{\mc{C}'}}]$ where the probability is taken over a random realization $\mc{C}'$ of $C$ and the randomization of the algorithm. 
		\item If $\gamma_{v,r-1} < c_v-2\epsilon^2$ call vertex $v$ {\em unsaturated} and {\em saturated} otherwise. 
		\item Construct a graph $H=(V_H, E_H)$ as described next. For every possible augmenting-hypewalk of size smaller than $2/\epsilon$ from $P$, we put a vertex in $V_H$ iff the first and last vertices in the walk are unsaturated. Moreover, we put an edge in $E_H$ between two nodes $u, v \in V_H$ iff their corresponding walks share at least a vertex.
		\item  \label{line:setI}$I \gets \apxMIS{H, \epsilon}$. \algcomment{This is an algorithm that returns an independent set of expected size at least $1-\epsilon$ fraction of some maximal independent set (MIS) of $H$.}
		\item $P' \gets P$.
		\item Iterate over all augmenting-hyperwalks $W \in I$ and apply them, i.e. set $P' \gets P' \Delta W$.
	\item Let $P'=((\mc{C}_0, M'_0), \ldots, (\mc{C}_\alpha, M'_\alpha))$ be the final profile. Return matching $M'_0$.
	\end{enumerate}
\end{tboxalgh}
We note a useful observation that essentially implies the entries of profile $P'$, which can be thought of as random variables of realization $\mc{C}$ and randomizations of the algorithm, are all drawn from the same distribution. The proof is essentially based on the fact that matchings $M_0, \ldots, M_\alpha$ are all drawn from the same distribution and treated symmetrically in algorithm, thus the resulting matchings $M'_0, \ldots, M'_\alpha$ all have the same distribution. See Section~\ref{sec:proofs} for a more formal proof.

\begin{observation}\label{obs:samedist}
	Matchings $M'_0, \ldots, M'_\alpha$ in profile $P'$ of algorithm $\findmatching{r}{\mc{C}}$ for any $r$ have the same distribution. That is, for any $i, j \in \{0, \ldots, \alpha\}$ and any matching $M'$ of $G$, $\Pr[M'_i = M'] = \Pr[M'_j = M']$.
\end{observation}

The algorithm operates only on the crucial edges and is thus clearly independent of the non-crucial edges and their realizations. Therefore, property 3 of Lemma~\ref{lem:independentmatching} is automatically satisfied. In what follows, we prove the other 3 properties in Sections~\ref{sec:p1}, \ref{sec:p2}, and \ref{sec:p4}.

\subsection{Lemma~\ref{lem:independentmatching} Property 1: The Matching's Size}\label{sec:p1}
In this section, we prove that algorithm $\findmatching{t}{\mc{C}}$  satisfies the first property of Lemma~\ref{lem:independentmatching}. That is the matching $Z$ returned by this algorithm satisfies $\E[|Z|] \geq q(C) - 30\epsilon \opt$.

Let us denote by $Z_r$ the matching returned by $\findmatching{r}{\mc{C}}$. Note that $Z_r$ is a random variable which is a function of both the randomization in realization $\mc{C}$ of $C$, and the internal randomizations used in algorithm $\findmatching{r}{\mc{C}}$.  (Observe that $Z = Z_t$.) Similarly, we define $P_r$, $H_r$, $I_r$, and $P'_r$ as the random variables referring to the values of $P$, $H$, $I$, and $P'$ in algorithm $\findmatching{r}{\mc{C}}$.

Property~1 of Lemma~\ref{lem:independentmatching} is a corollary  of  Lemma~\ref{lem:increaseexp} which states that for any $r$, if $\E[|Z_r|] \leq q(C)-30\epsilon\opt$, then $\E[|Z_r|] -\E[|Z_{r-1}|] \geq \epsilon^9\opt$.  Observe that it is sufficient for us as it implies that for any $r$ we have
$$
\E\big[|Z_t|\big] \geq \min\{q(C) - 30\epsilon \opt,  r\epsilon^9\opt\}.
$$ This gives us the desired result that $\E\big[|Z_t|\big]  \geq q(C) - 30\epsilon \opt$ for $t=1/\epsilon^9$, since $q(C) \leq \opt$. Below we state Lemma~\ref{lem:increaseexp} and prove it.
\begin{lemma}\label{lem:increaseexp}
	For any $r$, if $\E[|Z_r|] \leq q(C)-30\epsilon\opt$, then $\E[|Z_r|] -\E[|Z_{r-1}|] \geq \epsilon^9\opt$.
\end{lemma}
\begin{proof}[Proof outline]
This lemma is a direct result of Lemma~\ref{cl:increaseinZisI} and Lemma~\ref{lem:sizeI}. The first one states that for any $r$, we have $\E[|Z_r|] \geq \E[|Z_{r-1}|] + \frac{\E[|I_r|]}{\alpha+1}$ and the second one is that if $\E[|Z_{r-1}|] \leq q(C)-30\epsilon\opt$, then $\E[|I_r|] \geq 2\epsilon^2 \opt$. Combining these two lemmas gives us $\E[|Z_r|] -\E[|Z_{r-1}|] \geq \epsilon^9\opt$ and completes the proof as $\alpha=1/\epsilon^7 - 1$.
\end{proof}

\begin{lemma}\label{cl:increaseinZisI}
	For any $r$, it holds that $\E[|Z_r|] = \E[|Z_{r-1}|] + \frac{\E[|I_r|]}{\alpha+1}$.
\end{lemma}
\begin{proof} We start by proving that 
\begin{equation} \label{eq:md88} \sum_{v\in V} d_{P_r}(v) + 2|I_r| = \sum_{v\in V} d_{P'_r}(v).\end{equation} Note that, $P'_r$ is defined to be the result of iteratively applying all the augmenting hyperwalks of $I_r$ on $P_r$. Let $P^{(i)}_r$ be the result of iteratively applying the first $i$ augmenting hyperwalks of $I_r$ on $P_r$ and let $W_i$ be the hyperwalk that is to be applied in iteration $i$. We use proof by induction and show that for any $i$ we have $$\sum_{v\in V} d_{P_r}(v) + 2i = \sum_{v\in V} d_{P^{(i)}_r}(v).$$
Note that since hyperwalks in $I_r$ are vertex disjoint, for any two hyperwalks $W_1, W_2 \in I_r$ it holds that $W_2$ is an augmenting hyperwalk of $P_r\Delta W_1$ as well. This means that $W_i$ is indeed an augmenting hyperwalk of $P^{(i)}_r$. Moreover, recall that by definition of augmenting hyperwalks, after applying any augmenting hyperwalk on a profile $P$ there are only two vertices whose $d_{P}(v)$ increases by one and for the rest of the vertices it is unchanged. This gives us $$\sum_{v\in V} d_{P^{(i)}_r}(v) + 2 = \sum_{v\in V} d_{P^{(i+1)}_r}(v),$$ with completes the proof of  $$\sum_{v\in V} d_{P_r}(v) + 2|I_r| = \sum_{v\in V} d_{P'_r}(v)$$ since $P'_r = P^{|I_r|}_r$. 
Recall the definition $d_P(v) := \big |\{ i \mid v \in V(M_i) \}\big |$ for any profile $P$. Based on this definition, we can rewrite Equation~\ref{eq:md88} as 

\begin{equation} \label{eq:md99} \sum_{i=0}^{\alpha} |M_i| + |I_r| = \sum_{i=0}^{\alpha} |M'_i|.\end{equation} 

Observe that matchings $M_0, \dots, M_\alpha$ are coming from the same distribution and we have $\E\big[|M_i|\big] = \E\big[|Z_{r-1}|\big]$ for any $0\leq i\leq \alpha$. The reason is that they are the results of running the same matching algorithm on random realizations of $C$. Moreover, by Observation~\ref{obs:samedist}, matchings $M'_0, \dots, M'_\alpha$ are similarly coming from the same distribution which means for any $0\leq i\leq \alpha$ we have $Z_{r} = \E\big[|M'_0|\big] = \E\big[|M_i|\big]$. Combining this with Equation~\ref{eq:md99} we get
$$ \E\big[(\alpha+1)|Z_{r-1}| +|I|\big] = \E\left[\sum_{i=0}^{\alpha} |M_i| + |I|\right] = \E\left[\sum_{i=0}^{\alpha} |M'_i|\right] = \E\big[(\alpha+1)|Z_r|\big].$$
Dividing through by $\alpha + 1$ and rearranging the terms gives $ \E\big[|Z_{r-1}|\big] +\frac{\E[|I|]}{\alpha+1} =  \E\big[|Z_r|\big]$.
\end{proof}
Before proceeding to Lemma~\ref{lem:sizeI} and its proof we need the following definition.
\begin{definition}[Edge disjoint hyperwalks]
	We say two hyperwalks $W = ((e_1, s_1), \dots, (e_k, s_k) )$ and $W' = ((e'_1, s'_1), \dots, (e'_{k'}, s'_{k'}) )$ are edge disjoint if there does not exist indices $i< k$ and $j< k'$, where $e_i = e'_j$ and $s_i = s'_j$. 
\end{definition}

\begin{lemma} \label{lem:sizeI}
	If $\E[|Z_{r-1}|] \leq q(C)-30\epsilon\opt$, then $\E[|I_r|] \geq 2\epsilon^2 \opt$.
\end{lemma}

\begin{proof}
	To give the desired lower-bound for $\E[|I_r|]$ we first claim that if $\E[|Z_{r-1}|] \leq q(C)-30\epsilon\opt$, then there exists a set $O$ of edge-disjoint augmenting-hyperwalks of $P_r$ with unsaturated end-points where $\E[|O|] \geq8(\alpha+1)\epsilon\opt$. We later state this claim more formally in Lemma~\ref{lemma-size} and provide a  proof for it. We are interested in set $O$ for its two following properties. First, any hyperwalk in $O$ represents a node in graph $H_r$. Second, since the hyperwalks in $O$ are edge disjoint, any hyperwalk with length smaller than $2/\epsilon$ from $P_r$ can share vertices with at most $(\alpha+1)(2/\epsilon)$ hyperwalks in this set. We note that $(\alpha+1)$ is the maximum number of edge disjoint hyperwalks that can pass through a single vertex. Combining these two properties gives that the expected size of any maximal independent set of $H_r$ is at least $\E\big[|O|\big]/(2(\alpha+1)/\epsilon) = 4\epsilon^2\opt$ since there is an edge between two vertices in $H_r$ iff their corresponding hyperwalks share at least a vertex. As stated in Line~\ref{line:setI} of \findmatching{r}{\mc{C}}, set $I_r$ is an independent set of $H_r$ with size at least $(1-\epsilon)$ fraction of a maximal independent set of $H_r$. Therefore, we have $$\E\big[|I_r|\big] \geq 4(1-\epsilon)\epsilon^2\opt.$$ Assuming that $\epsilon\leq 1/2$ we complete the proof of this claim and obtain $\E\big[|I_r|\big] \geq 2\epsilon^2\opt$.
	\end{proof}
	
In the rest of this section we focus on proving the following lemma which is previously used to complete the proof of Lemma~\ref{lem:sizeI}. Since the proof is detailed and consists of independent arguments, it includes two claims that are needed to complete the proof.
 
\begin{lemma} \label{lemma-size}
For any $r\in [t]$, if $\E\big[|Z_{r-1}|\big] \leq q(C) - 30 \epsilon \opt$, then there exists a set $O$ of edge-disjoint augmenting hyperwalks of profile $P_r =  ((\mc{C}_0, M_0), \ldots, (\mc{C}_\alpha, M_\alpha))$ with unsaturated endpoints where $\E[|O|] \geq8(\alpha+1)\epsilon\opt.$ 
\end{lemma}

We will first construct set $O$ and then give a  lower-bound for its expected size. Draw $\alpha+1$ realizations $\mc{N}_0, \dots, \mc{N}_\alpha$ of the non-crucial graph $N$. For any $0\leq i \leq \alpha$, let $M^g_i := \MM{\mc{N}_i \cup \mc{C}_i}$ where $M$ returns a unique maximum matching that was also used in Algorithm~\ref{alg:sampling}. Call an edge of graph $\mc{C}_i$ \emph{green} iff it is in matching $M^g_i$ but not in matching $M_i$. Alternatively, we call an edge \emph{red} iff it is in $M_i$ but not in $M^g_i$. 
To construct set $O$ we give an algorithm to iteratively find hyperwalks that alternate between green and red edges. Since we need our hyperwalks to be edge-disjoint, after using an edge of a subgraph we mark it as used and ignore it for the rest of the algorithm. 

At each iteration of the algorithm, we construct a hyperwalk $W$ as follows until there is no such a hyperwalk left. Pick an unsaturated vertex $v$ and a subgraph $\mc{C}_i$ such that $v$ has an unused green edge in $\mc{C}_i$ but not a red one. Denote this green edge by $e=(v, v')$ and choose $(e, i)$ to be the first element of our hyperwalk. If vertex $v'$ has a red edge $e'$ in subgraph $\mc{C}_i$ we add $(e', i)$ to our hyperwalk, otherwise we look for a  subgraph $\mc{C}_j$ in which $v'$ has an unused red edge $e'$ but not a green one and choose $(e', j)$ as the second element of the hyperwalk. We continue this process by alternating the colors until it is not possible to continue. Let $u$ be the vertex in which our hyperwalk ends. If $u$ is saturated we add $W$ to a set $T_2$. Otherwise, if the last edge of $W$ is green we add it to $O$ and if it is red we add $W$ to $T_1$. In the following claim we show that the hyperwalks in $O$ have the desired property and we later prove that $|O|$ is large enough.

\begin{claim}
Any $W \in O$ is an augmenting-hyperwalks that begins and ends in unsaturated vertices. 
\end{claim}

\begin{proof}
Any hyperwalk in $O$ begins with an unsaturated vertex and ends in one. Also, hyperwalks in $O$ are edge disjoint since after adding an element $(e, i)$ to a hyperwalk we mark $e$ as used in subgraph $\mc{C}_i$ and do not add it to other hyperwalks. It only remains to prove that every hyperwalk $W = ((e_1, s_1), \dots, (e_k, s_k))\in O$ is indeed an augmenting-hyperwalk.

Let $P_r \Delta W$ be the result of applying $W$ on $P_r = ((\mc{C}_0, M_0), \ldots, (\mc{C}_\alpha, M_\alpha))$. By Definition~\ref{def:aug}, there are three conditions that $P_r \Delta W$ should satisfy if $W$ is an augmenting-hyperwalk. The first condition is that any $M'_i$ is a matching in $\mc{C}_i$ where
$$
	M'_i = M_i \cup \{ e_j \mid j \text{ is odd, and } s_j = i\} \setminus \{ e_j \mid j \text{ is even, and } s_j = i\}, \qquad \text{for all $i \in \{0, \ldots, \alpha\}$}.
$$
Note that $W$ is alternating between green and red edges with green ones being in the odd positions. Further, for any element $(e, i)$ in an odd position $j$ and any red edge $e'$ adjacent to it in $\mc{C}_i$, hyperwalk $W$ contains $(e', i)$ in either position $j-1$ or position $j+1$; thus the first condition is satisfied.

As for the second condition, since $W$ is alternating between green and red edges applying it would satisfy $d_{P_r}(v) = d_{P_r \Delta W}(v)$ for any vertex $v$ that is not an end-point. Moreover, $P_r \Delta W$ simply satisfies the third condition that is $d_{P_r}(v) + 1 = d_{P_r \Delta W}(v)$ iff $v$ is the first or the last vertex of the hyper-walk since $W$ begins and ends with green edges.
\end{proof}

To complete the proof of Lemma~\ref{lemma-size}, we need to show that $\E[|O|] \geq 8(\alpha+1)\epsilon\opt$. For any vertex $v$, let $g_{v, i}$ be the number of subgraphs $\mc{C}_0, \dots \mc{C}_\alpha$ in which $v$ has an unused green edge after the $i$-th iteration of the algorithm and similarly define $r_{v, i}$ to be the number of subgraphs in which $v$ has an unused red edge after the $i$-th iteration. Each iteration here  means constructing a hyperwalk and marking its edges as used. Also, let us respectively denote the set of saturated and unsaturated vertices by $S$ and $U$. 
Consider the hyperwalk $W_i$ constructed in the $i$-th iteration. Observe that if $W_i\in O$, we have $$\sum_{v\in U} (g_{v, i-1} - r_{v, i-1}) - \sum_{v\in U} (g_{v, i} - r_{v, i}) =2 $$ since any hyperwalk in $O$ starts from an unsaturated vertex with a green edge and ends the same way. However, if $W_i\in T_2$, we have $$\sum_{v\in U} (g_{v, i-1} - r_{v, i-1}) - \sum_{v\in U} (g_{v, i} - r_{v, i}) = 1,$$ and if $W_i\in T_1$ we have $$\sum_{v\in U} (g_{v, i-1} - r_{v, i-1}) - \sum_{v\in U} (g_{v, i} - r_{v, i}) = 0.$$
We claim that when our algorithm stops after $j$ iterations $\sum_{v\in U} (g_{v, j} - r_{v, j}) \leq 0$ holds. This is because otherwise, we could still find a subgraph $\mc{C}_i$ and a vertex $v$ where $v$ has a green edge in $\mc{C}_i$ but not a red one and start a new hyperwalk. As a result we have the following lower-bound for $|O|$, where for brevity, in the rest of the proof we use $g_{v}$ and $r_{v}$ instead of $g_{v,0}$ and $r_{v,0}$:

$$|O| \geq  \frac{1}{2} \left(\sum_{v\in U} (g_{v} - r_{v}) - |T_2|\right).$$
Taking expectations,
\begin{equation} \label{eq:m22}
\E\big[|O|\big] \geq \frac{1}{2}\E\left[\sum_{v\in U} (g_{v} - r_{v}) - |T_2|\right] = \frac{1}{2}\E\left[\sum_{v\in U} (g_{v} - r_{v})\right] - \frac{1}{2}\E\big[|T_2|\big].
\end{equation}
We first focus on bounding $\E\big[\sum_{v\in U} (g_{v} - r_{v})\big]$ and prove that it is upper-bounded by $40\alpha\epsilon\opt$.

$$\sum_{v\in U} \E[g_{v} - r_{v}\big] =
 \sum_{v\in V} \E[g_{v} - r_{v}\big]  - \sum_{v\in S} \E[g_{v} - r_{v}\big] $$
 Note that $c_v$, by definition, is the probability with which vertex $v$ is matched in any $M^g_i$. Moreover, $\gamma_{v, r}$ is the probability with which vertex $v$ is matched in any $M_i$ which means $\E[g_v - r_v]= (\alpha+1)(c_v-\gamma_{v, r})$ and
 $$
 \E\left[\sum_{v\in V} (g_v-r_{v})\right] = 2(\alpha+1) (q(C) - \E\left[|Z_r|\right]).
 $$ 
 Also, since $ \E\big[|Z_r|\big] \leq q(C) - 30 \epsilon \opt$ we obtain
 $$
 \E\left[\sum _{v\in V} g_{v}\right] - \E \left[\sum_{v\in V}r_{v}\right] \geq 60(\alpha+1)\epsilon\opt.
 $$
 Moreover, by definition of saturated vertices, we know that $c_v-\gamma_{v, r} \leq 2\epsilon^2$ holds for any saturated vertex $v$ which results in 
 
 \begin{equation} \label{eq:m-exp}\sum_{v\in S} \E[g_{v} - r_{v}\big] \leq 2n(\alpha+1)\epsilon^2 \leq 20(\alpha+1)\epsilon \opt.\end{equation} Note that $2n(\alpha+1)\epsilon^2  \leq 20(\alpha+1)\epsilon \opt$ comes from Assumption~\ref{ass:optlarge} that $\opt\geq 0.1\epsilon n$. Combining these equation, we get
 
\begin{equation}\label{eq:m33} \sum_{v\in U} \E[g_{v} - r_{v}\big]  \geq 40(\alpha+1) \epsilon \opt.   \end{equation}
In the next step, we provide an upper-bound for $\E\big[|T_2|\big]$ and to do so we first prove the following claim.

\begin{claim}
For any vertex $v \in S$ the number of hyperwalks in $T_2$ that end in $v$ is $\leq |g_v-r_v|$. 
\end{claim}
\begin{proof}
Consider the hyperwalk $W$ that is the first one to be constructed among the hyperwalks in set $T_2$ that end in vertex $v$ and let $(e, i)$ be its last element. W.l.o.g., assume that the color of edge $e$ in graph $\mc{C}_i$ is red.
The fact that $W$ stops in vertex $v$ means that at the time of construction of this hyperwalk, there is no subgraph $\mc{C}_j$ that has an unused green edge of $v$ but not a red one. Therefore, from this point of the algorithm, any subgraph $\mc{C}_k$ that contains an unused green edge $e_g$ of vertex $v$ also has an unused red edge $e_r$ of this vertex. We note that based on our algorithm if a hyperwalk with last element $(e, i)$ stops at vertex $v$  then subgraph $\mc{C}_i$ either does not contain a green edge of $v$ or a red edge of this vertex. Moreover, due to the fact that $W$ is the first hyperwalk to stop in vertex $v$ we know that previously constructed hyperwalks contain the same number of green and red edges of vertex $v$. This means that there are at most $|g_v-r_v|$ many possibilities for the last element of a hyperwalk that stops at $v$ and since our hyperwalks are edge disjoint then for any vertex $v \in S$ the number of hyperwalks in $T_2$ that end in $v$ is upper-bounded by $|g_v-r_v|$.
\end{proof}

Based on the aforementioned claim, the number of hyperwalks ending in saturated vertices is at most $\sum_{v\in S} \E[g_v - r_v]$, which means $\E\big[|T_2|\big] \leq \sum_{v\in S} \E\big[|g_v - r_v|\big]$, implying further that
\begin{align}
\nonumber \E[|T_2|] &\leq \sum_{v\in S} \E\big[|g_v - r_v|\big]\\
\nonumber	&= \sum_{v\in S} \E\Big[|g_v - \E[g_v] - r_v + \E[r_v]+ \E[g_v] - \E[r_v]|\Big]\\
\nonumber &\leq \sum_{v\in S} \E\Big[|g_v - \E[g_v]| + |r_v - \E[r_v]|+ |\E[g_v] - \E[r_v]|\Big]\\
&\leq \sum_{v\in S} \left(\E\big[|g_v - \E[g_v]|] + \E[|r_v - \E[r_v]|]\right)+ \sum_{v\in S} (\E[g_v] - \E[r_v]). \label{eq:812301982355}
\end{align}
The last equation is due to the fact that $\E[g_v] \geq \E[r_v]$ for all $v$.

Using a simple application of Chebyshev's inequality, we show that for any vertex $v$, we have  $\E[|r_v - \E[r_v]|] \leq 2(\alpha+1)^{2/3}$ and $\E[|r_g - \E[r_g]|] \leq 2(\alpha+1)^{2/3}$. Note that  we have $\var(g_v) \leq \alpha+1$ and  $\var(g_v) \leq \alpha+1$. Using Chebyshev's inequality, we have $\Pr[|r_v - \E[r_v]|\geq \beta (\alpha +1)^{1/2}] \leq 1/\beta^2.$ By setting $\beta = (\alpha +1)^{1/6}$, we get  $\Pr[|r_v - \E[r_v]|\geq (\alpha +1)^{2/3}] \leq (\alpha +1)^{-1/3},$ which gives us $\E[|r_v - \E[r_v]|] \leq 2 (\alpha +1)^{2/3}.$ Similarly, we have $\E[|g_v - \E[g_v]|] \leq 2 (\alpha +1)^{2/3}$. As a result, we get 
$$\sum_{v\in S}(\E[|r_v - \E[r_v]|] + \E[|r_g - \E[r_g]|]) \leq 4n(\alpha+1)^{2/3}.$$ Since in  \findmatching{r}{\mc{C}}  we set $\alpha=1/\epsilon^7-1$ and since $n\leq 10\opt/\epsilon$ we have

$$
\sum_{v\in S}(\E[|r_v - \E[r_v]|] + \E[|r_g - \E[r_g]|]) \leq \frac{40\opt (1/\epsilon^7)^{2/3}}{\epsilon} = \frac{40 \opt}{\epsilon^{14/3} \times \epsilon} = \frac{40\opt}{\epsilon^{7} \times \epsilon^{-4/3}} = 40(\alpha+1)\epsilon^{4/3} \opt.
$$ 
Moreover, by (\ref{eq:m-exp}) we have $$\sum_{v\in S}(\E[g_v] - \E[r_v]) \leq 20(\alpha+1)\epsilon\opt.$$
Combining these two bounds into (\ref{eq:812301982355}) we get
\begin{equation} \label{eq:m44}
\sum_{v\in S} \E\big[|T_2|\big] \leq (\alpha+1)\epsilon\opt(20+40\epsilon^{1/3}).
\end{equation}
Incorporating (\ref{eq:m33}) and (\ref{eq:m44}) into (\ref{eq:m22}) and simplifying, gives
\begin{equation}
\E\big[|O|\big] \stackrel{\text{(\ref{eq:m22})}}{\geq} \frac{1}{2}\E\Big[\sum_{v\in U} (g_{v} - r_{v})\Big] - \frac{1}{2}\E\big[|T_2|\big] \stackrel{\text{(\ref{eq:m33}), (\ref{eq:m44})}}{\geq} (\alpha+1)\epsilon\opt(10-20\epsilon^{1/3}). \end{equation} 
By letting $\epsilon$ be small enough, we can assume that $\epsilon^{1/3} \leq 0.1$ and get  $$\E[|O|]\geq 8(\alpha+1)\epsilon\opt,$$
which completes the proof of Lemma~\ref{lemma-size}. This completes all the components needed within the proof of Lemma~\ref{lem:increaseexp} which as discussed at the start of the section, implies the needed bound on the expected size of the matching returned.

\subsection{Lemma~\ref{lem:independentmatching} Property 2: Matching Probabilities}\label{sec:p2}
In this section, we prove that algorithm $\findmatching{t}{\mc{C}}$ satisfies property 2 of Lemma~\ref{lem:independentmatching} that for each vertex $v$, $\Pr[X_v] \leq \max\{c_v - \epsilon^2, 0\}$. Recall that $X_v$, as defined in Lemma~\ref{lem:independentmatching}, is the indicator of the event that $v$ is matched in \findmatching{t}{\mc{C}}, and the probability is taken over both the realization $\mc{C}$ and the randomization of algorithm $\findmatching{t}{\mc{C}}$.

Let us use $X_{v, r}$ to denote the event that vertex $v$ gets matched in matching $\findmatching{r}{\mc{C}}$. It holds that $X_{v, t} = X_v$. Therefore, it suffices to show that $\Pr[X_{v, t}] \leq \max\{c_v - \epsilon^2, 0\}$. We will, however, prove a stronger claim:

\begin{claim}\label{cl:p2holds}
	For every integer $r$ and for every vertex $v$, it holds that $\Pr[X_{v, r}] \leq \max\{c_v - \epsilon^2, 0\}$.
\end{claim}

We prove this by indiction on $r$. For the base case $r=0$, algorithm $\findmatching{0}{\mc{C}}$ returns an empty matching $\emptyset$. Therefore $\Pr[X_{v, 0}] = 0$ for all vertices $v$, clearly satisfying the claim. For the induction step, fix any vertex $v$. We suppose that $\Pr[X_{v, r-1}] \leq \max\{c_v - \epsilon^2, 0\}$ and prove that it continues to hold that $\Pr[X_{v, r}] \leq \max\{c_v - \epsilon^2, 0\}$. We start with a definition.

\begin{definition}
Define $\rho_v$ to be the fraction of matchings $M_0, \ldots, M_\alpha$ in which $v$ is matched and define $\rho'_v$ similarly with respect to matchings $M'_0, \ldots, M'_\alpha$. More precisely,
$$
	\rho_v := \frac{|\{i : v \in V(M_i) \}|}{\alpha + 1}, \qquad \text{ and } \qquad \rho'_v := \frac{|\{i : v \in V(M'_i) \}|}{\alpha + 1}.
$$	
\end{definition}

\begin{observation}
	$\E[\rho_v] = \Pr[X_{v, r-1}]$ and $\E[\rho'_v] = \Pr[X_{v, r}]$.
\end{observation}
\begin{proof}
	For any $i \in \{0, \ldots, \alpha\}$, we have $\Pr[v \in V(M_i)] = \Pr[X_{v, r-1}]$ since $M_i = \findmatching{r-1}{\mc{C}_i}$ and $\mc{C}_i$ is picked from the same distribution that the actual realization $\mc{C}$ is picked from. Thus: 
	$$
		\E[\rho_v] = \E\left[\frac{\sum_{i=0}^\alpha \mathbbm{1}(v \in V(M_i))}{\alpha+1}\right] = \frac{1}{\alpha+1} \sum_{i=0}^\alpha \Pr[v \in V(M_i)] = \frac{1}{\alpha+1} \sum_{i=0}^\alpha \Pr[X_{v, r-1}] = \Pr[X_{v, r-1}].
	$$
	For the second equality, first observe that since $M'_0$ is the matching returned by $\findmatching{r}{\mc{C}}$, then $X_{v, r}$ is by definition exactly the event that $v \in V(M'_0)$ and thus $\Pr[X_{v, r}] = \Pr[v \in V(M'_0)]$. Moreover, due to symmetry of the algorithm in constructing $M'_0, \ldots, M'_\alpha$, it holds for any $i \in [\alpha]$ that $\Pr[v \in V(M'_i)] = \Pr[v \in V(M'_0)] = \Pr[X_{v, r}]$. Therefore, we get:
	$$
		\E[\rho'_v] = \E\left[\frac{\sum_{i=0}^\alpha \mathbbm{1}(v \in V(M'_i))}{\alpha+1}\right] = \frac{1}{\alpha+1} \sum_{i=0}^\alpha \Pr[v \in V(M'_i)] = \frac{1}{\alpha+1} \sum_{i=0}^\alpha \Pr[X_{v, r}] = \Pr[X_{v, r}],
	$$
	concluding the proof.
\end{proof}

In algorithm \findmatching{r}{\mc{C}}, we mark $v$ as either saturated or unsaturated depending on the value of $\gamma_{v, r-1}$. Note from definition of $\gamma_{v, r-1}$ that $\gamma_{v, r-1} = \Pr[X_{v, r-1}]$. Therefore, $v$ is marked as saturated if $\Pr[X_{v, r-1}] \geq c_v - 2\epsilon^2$ and unsaturated if $\Pr[X_{v, r-1}] < c_v - 2\epsilon^2$. We consider the two cases individually.

\smparagraph{If $v$ is saturated.} In this case, by definition of graph $H$, vertex $v$ cannot start or end any augmenting-hyperwalk with a corresponding vertex in $H$ (and for that matter in $I$). By definition of augmenting-hyperwalks, for all vertices (except the endpoints of the walk) applying the hyperwalk does not change the number of matchings in which the vertex is part of. Therefore, if $v$ is saturated, $\rho_v = \rho'_v$ and thus $\Pr[X_{v, r}] = \Pr[X_{v, r-1}] \leq \max\{c_v - \epsilon^2, 0\}$ where the latter inequality comes from the induction's hypothesis.

\smparagraph{If $v$ is unsaturated.} Note that in graph $H$ by definition we have edges between any pair of augmenting-hyperwalks that share a vertex in the graph. Therefore, the independent set $I$ of $H$ can include at most one augmenting-hyperwalk $W$ that includes vertex $v$. If $v$ is not an end-point of $W$, then as in the case above, we get $\rho_v = \rho'_v$. However, if $v$ is an end-point of $W$, then by definition of augmenting-hyperwalks, there will be one (and only one) $i$ where $v \in V(M'_i)$ and $v \not\in V(M_i)$. In this case, we get that
$$
	\rho'_v = \frac{|\{i : v \in V(M'_i) \}|}{\alpha + 1} = \frac{|\{i : v \in V(M_i) \}| + 1}{\alpha + 1} = \rho_v + \frac{1}{\alpha+1} \stackrel{\alpha = 1/\epsilon^3-1}{<} \rho_v + \epsilon^2.
$$
Since in this case, we had $\rho_v < c_v - 2\epsilon^2$, we get $\rho'_v < c_v - 2\epsilon^2 + \epsilon^2 = c_v -\epsilon^2$. Therefore the induction's hypothesis still holds that $\Pr[X_{v, r}] < \max\{c_v - \epsilon^2, 0\}$, completing the proof of Claim~\ref{cl:p2holds}.

\subsection{Lemma~\ref{lem:independentmatching} Property 4: Matching Independence}\label{sec:p4}
In this section, we prove that algorithm $\findmatching{t}{\mc{C}}$ satisfies property 4 of Lemma~\ref{lem:independentmatching}. That is, for every subset $I= \{v_1, \ldots, v_k\}$ of the vertices such that $d_C(v_i, v_j) \geq \lambda$ for all $v_i, v_j \in I$, random variables $X_{v_1}, \ldots, X_{v_k}$ are independent. Recall that $X_v$ for a vertex $v$ is the indicator of the event that $v$ is matched in the matching returned by \findmatching{t}{\mc{C}}. We also, again, emphasize that this ``independence'' is with regards to the randomization of realization $\mc{C}$ of $C$ on which $Z$ is constructed, and the randomization of algorithm $\findmatching{t}{\mc{C}}$ itself.

In Section~\ref{sec:crucialoverview} we gave an overview of how we can argue about such independence via an implementation of the algorithm in the \local{} model of computation. Here we give this implementation.

\smparagraph{Initialization.} The communication network is graph $C$. Each node $v$ is initially given the following information: Its incident edges in $C$ and how they are realized, the maximum degree $\Delta_C$ of graph $C$, parameter $\epsilon$, and the value of $c_v$. Note that to gather information about realization of edges further away, the nodes need to communicate. Also note that even though the value of $c_v$ may reveal some information about graph $G$ (or $C$), it crucially reveals no information about the realization $\mc{C}$ of $C$, or other sources of randomization used by the algorithm. Thus, property 4 can still be satisfied if we manage to show the algorithm can be implemented in few rounds.

\smparagraph{The \apxMIS{H, \epsilon} algorithm.} First, we mention that subroutine \apxMIS{H, \epsilon} already has an efficient \local{} implementation whose round-complexity depends only on the maximum degree of $H$ and $\epsilon$, without essentially any dependence on the number of nodes in $H$. Any implementation with such round-complexity can be used in our case. For instance, we use one implied in \cite{DBLP:conf/soda/Ghaffari19} (see Appendix~\ref{app:weakmis} for details):

\begin{lemma}[\cite{DBLP:conf/soda/Ghaffari19}]\label{lem:apxMIS}
	Given a graph $H$ of max degree $\Delta$ and any parameter $\epsilon$, there is a \local{} algorithm $\apxMIS{H, \epsilon}$ that returns an independent set $I$ of $H$ in $O(\log \frac{\Delta}{\epsilon})$ rounds such that the expected size of $I$ is at least $(1-\epsilon)$ fraction of some maximal independent set of $H$.
\end{lemma}

We give a \local{} implementation of Algorithm~\ref{alg:crucial} which proves the following:

\begin{claim}\label{cl:localimp}
	For any $r \geq 0$, algorithm $\findmatching{r}{\cdot}$ can be implemented in $O(r \epsilon^{-4}\log{\Delta_C})$ rounds of \local{}.
\end{claim}

\begin{proof}
We prove the claim by induction on $r$. For the base case, algorithm \findmatching{0}{\cdot} can be implemented in 0 rounds since the output is always the empty matching. We assume that algorithm \findmatching{r-1}{\cdot} can be implemented in $\beta(r-1)\epsilon^{-4}\log \Delta_C$ rounds where $\beta > 1$ is a sufficiently large absolute constant that we fix later, and prove that $\findmatching{r}{\cdot}$ can be implemented in $\beta r\epsilon^{-4}\log \Delta_C$ rounds. 

\textbf{Step 1.} First, the algorithm draws $\alpha$ realizations $\mc{C}_1, \ldots, \mc{C}_\alpha$. Since information about realization of edges is stored locally on their incident vertices, we can easily generate these random realizations in $O(1)$ rounds. After that, on each graph $\mc{C}_i$ for $i \in \{0, \ldots, \alpha\}$, we recursively run the $(\beta(r-1)\epsilon^{-4}\log \Delta_C)$-round implementation of $\findmatching{r-1}{\mc{C}_i}$. Note that all of these can run in parallel. The overall round-complexity of this step, is thus $\beta(r-1)\epsilon^{-4}\log \Delta_C + O(1)$.

\textbf{Step 2.} Next, we need to compute $\gamma_{v, r-1}$ for each vertex $v$, which recall is the probability that $v$ is matched in $\findmatching{r-1}{\mc{C'}}$ where $\mc{C}'$ is a random realization of $C$. The crucial observation here is that since $\findmatching{r-1}{\cdot}$ can, by the induction hypothesis, be implemented within only $\beta(r-1)\epsilon^{-4}\log \Delta_C$ rounds, $\gamma_{v, r-1}$ is merely a function of the topology induced in the $(\beta(r-1)\epsilon^{-4}\log \Delta_C)$-hop of $v$. We first gather this neighborhood of $v$, which can be done in $(\beta(r-1)\epsilon^{-4}\log \Delta_C)$ rounds, then compute $\gamma_{v, r-1}$. We note that this gathering part can be done in parallel to the operations of Step 1. Therefore, overall, Steps 1 and 2 take $(\beta(r-1)\epsilon^{-4}\log \Delta_C + O(1))$ rounds. Having $\gamma_{v, r-1}$ for each vertex $v$, we can then determine for each vertex whether it is saturated or unsaturated since we are given the value of $c_v$ in the initialization step. 

\textbf{Step 3.} The next step is constructing graph $H$. In graph $H$, each vertex corresponds to a walk of size at most $2/\epsilon$ in $C$. Therefore, each vertex in $C$ can first gather all such walks around it in $O(1/\epsilon)$ rounds, and then determine which one of them are augmenting-hyperwalks satisfying the required properties to be considered as a node of $H$. Determining the edges of $H$ can also be done locally; once we construct the vertices, there will be an edge between any two walks that share a vertex. Therefore, overall, graph $H$ can be constructed in $O(1/\epsilon)$ rounds.

\textbf{Step 4.} Once we construct $H$, we run the \local{} implementation of $\apxMIS{H, \epsilon}$ mentioned in Lemma~\ref{lem:apxMIS} on graph $H$. We emphasize that our communication network here is graph $C$, not $H$. However, any message between two nodes of $H$ can be sent over network $C$ within $O(1/\epsilon)$ rounds. This is because any two incident nodes of $H$, are walks of size at most $O(1/\epsilon)$ in $C$ that share at least a vertex. The overall running time of this procedure is thus  $O(\frac{1}{\epsilon} \times \log \frac{\Delta_H}{\epsilon})$. We note that $\Delta_H = O(\epsilon^{-1} ((\alpha+1)\Delta_C)^{2/\epsilon})$. To see this, fix any walk $w$ with a corresponding node in $H$. This walk has at most $2/\epsilon$ nodes in $C$. Now each node in $C$ is incident to $O(((\alpha+1)\Delta_C)^{2/\epsilon})$ hyperwalks: There are $O((\Delta_C)^{2/\epsilon})$ walks of size $\leq 2/\epsilon$ branching out of each of the nodes, and each edge of the walk can take on $\alpha+1$ labels from $\{0, \ldots, \alpha\}$ to be transformed to a hyperwalk. Therefore, overall the number of rounds required for this part of the algorithm is
$$
	O\left(\frac{1}{\epsilon} \times \log \frac{\Delta_H}{\epsilon}\right) = O\left(\frac{1}{\epsilon} \log \frac{\epsilon^{-1} ((\alpha+1)\Delta_C)^{2/\epsilon}}{\epsilon} \right) = O\left(\epsilon^{-4} \log \Delta_C \right),
$$
where the last equality comes from the fact that $\alpha = \poly(\epsilon^{-1})$.

\textbf{Step 5.} Finally, applying the augmenting-hyperwalks chosen in $I$ is simple and can be done in $O(1/\epsilon)$ rounds since these walks are of size $\leq 2/\epsilon$.

\textbf{Round-complexity.} Let $\beta_2$ be a sufficiently large constant by multiplying which we can surpass the $O$-notations. We get
\begin{flalign*}
\text{\# of rounds} &\leq 
\underbrace{\beta(r-1)\epsilon^{-4}\log \Delta_C + \beta_2}_{\text{Steps 1 and 2}} + \underbrace{\beta_2(1/\epsilon)}_{\text{Step 3}} + \underbrace{\beta_2(\epsilon^{-4}\log\Delta_C)}_{\text{Step 4}} + \underbrace{\beta_2(1/\epsilon)}_{\text{Step 5}}\\
 &< \beta(r-1)\epsilon^{-4}\log \Delta_C + 4\beta_2(\epsilon^{-4}\log\Delta_C)\\
	&= \Big( \beta(r-1)+4\beta_2 \Big)\epsilon^{-4}\log\Delta_C.
\end{flalign*}
Since $\beta_2$ is an absolute constant that does not depend on $\beta$, we can set $\beta$ to be large enough with respect to it. Setting $\beta = 4\beta_2$ is sufficient since
$$
	\Big( \beta(r-1)+4\beta_2 \Big)\epsilon^{-4}\log\Delta_C = \beta r \epsilon^{-4}\log \Delta_C.
$$
This concludes the proof of the induction step, and consequently the proof of Claim~\ref{cl:localimp}.
\end{proof}

We showed in Claim~\ref{cl:localimp} that algorithm \findmatching{r}{\mc{C}}, for any $r$, can be implemented within $O(r\epsilon^{-4}\log \Delta_C)$ rounds of \local{}. Our final algorithm for Lemma~\ref{lem:independentmatching} is $\findmatching{t}{\mc{C}}$ where we set $t = 1/\epsilon^9$. Thus, the output of each vertex can be determined within $\lambda'=O(\epsilon^{-13}\log\Delta_C)$ rounds. This, as described, proves property 4 of Lemma~\ref{lem:independentmatching} since $\lambda=\epsilon^{-20}\log \Delta_C$ is larger than $\lambda'/2$ given that $\epsilon$ is small enough to surpass the hidden constants in the $O$-notation. (Recall that we can assume $\epsilon$ is smaller than any needed constant.)

\section{Deferred Proofs}\label{sec:proofs}

\begin{proof}[Proof of Lemma~\ref{lem:gap}]
	Let $t_0 = (\epsilon p)^{50}$ and for any $i \geq 1$ let $t_i = f(t_{i-1})$. Note that $t_0 > t_1 > t_2 > \ldots$ by the assumption of the lemma that $0 < f(x) < x$ for all $0 < x < 1$. For any $i \geq 1$ define $q_i = \sum_{e \in E: q_e \in (t_i, t_{i-1}]} q_e$ and let $j$ be the smallest number where $q_j \leq \epsilon \opt$. We will soon prove existence of such $j$ and also prove that $j = O(1/\epsilon)$. We claim that setting $\tau_+ = t_{j-1}$ and $\tau_- = t_j$ satisfies the conditions of the lemma.
	
	\textbf{Condition} (1): This condition holds trivially since $\tau_- = t_{j} = f(t_{j-1}) = f(\tau_+)$. 
	
	\textbf{Condition} (2): Let us define $X := \{ e \mid \tau_- < q_e < \tau_+ \}$. Recall that crucial and non-crucial edges are defined based on $\tau_+$ and $\tau_-$. That is, an edge $e$ is crucial (i.e. $e \in C$) if $q_e \geq \tau_+$, and is non-crucial (i.e. $e \in N$) if $q_e \leq \tau_-$. This implies that the remaining edges that are neither crucial nor non-crucial belong to $X$. Therefore,
	$$
		\opt = q(E) = q(C) + q(N) + q(X).
	$$
	To obtain $q(N)+q(C) \geq (1-\epsilon)\opt$ it thus suffices to show $q(X) \leq \epsilon \opt$. Noting that $\tau_+ = t_{j-1}$ and $\tau_- = t_j$ and also noting the definition of $q_j$ above, we get $q(X) \leq q_j$. Recall that we chose $j$ such that $q_j \leq \epsilon \opt$. Therefore we indeed get that $q(X) \leq \epsilon \opt$.
	
	\textbf{Condition} (3): We defined $t_0 = (\epsilon p)^{50}$ and recursively defined $t_i = f(t_{i-1})$. Since $f(\cdot)$ is only a function of its input, we get via a simple induction that both $t_j$ and $t_{j-1}$ are also functions of only $\epsilon$ and $p$. (Recall that $j = O(1/\epsilon)$.)
	
	\textbf{Condition} (4): We defined $t_0 = (\epsilon p)^{50}$ and recall that we showed $t_0 > t_1 > t_2  > \ldots$; this implies clearly that $\tau_+ = t_{j-1} \leq (\epsilon p)^{50}$.
	
	\textbf{Existence of $j$.} It only remains to prove that there exists a choice of $j$ satisfying $q_j \leq \epsilon \opt$ and that this $j$ is not too large. Precisely, we show that $j = O(1/\epsilon)$. Since intervals $(t_1, t_0], (t_2, t_1], (t_3, t_2], \ldots$ are disjoint, it holds that for each edge $e$ there is at most one $i$ for which $q_e \in (t_i, t_{i-1}]$. This means that $\sum_{i=1}^\infty q_i \leq \sum_{e \in E} q_e = \opt$.  It thus has to hold that $j \leq \lceil 1/\epsilon \rceil + 1$ or otherwise 
	$$
	\sum_{i=1}^{j-1} q_i \geq \sum_{i=1}^{\lceil 1/\epsilon \rceil + 1} \epsilon \opt = (\lceil 1/\epsilon \rceil + 1) \epsilon \opt >  \opt
	$$ contradicting the previous statement. This concludes the proof of the lemma.
\end{proof}

\begin{proof}[Proof of Claim~\ref{cl:frange}]
	We prove parts 1-3 one by one.
	
	\smparagraph{Part 1.} The upper bound $\E[f_e] \leq q_e$ is simple to prove. Consider random variable $f'_{e} = t_e/R$ and note that $f'_e \geq f_e$. We have
	$$
	\E[f'_e] = \E\left[\frac{t_e}{R}\right] = \frac{1}{R} \E[t_e] = \frac{1}{R} \left(\sum_{i=1}^R \Pr[e \in \MM{\mc{G}_i}]\right) = \frac{1}{R} (R \times \Pr[e \in \MM{\mc{G}_1}]) = q_e.
	$$
	Since $f_e \leq f'_e$, we get $\E[f_e]\leq\E[f'_e] = q_e$, concluding the proof of part 1.
	
	\smparagraph{Part 2.} Next we turn to prove the lower bound $\E[f_e] \geq (1-\epsilon)q_e$. Let $X_i$ be the indicator random variable for $e \in \MM{\mc{G}_i}$. We have $t_e = X_1 + \ldots + X_R$, $\E[X_i] = q_e$, and $\E[t_e] = Rq_e$. Note also that the $X_i$'s are independent since graphs $\mc{G}_1, \ldots, \mc{G}_R$ are drawn independently. Therefore, $\Var[t_e] = \sum_{i=1}^R \Var[X_i] = R(q_e - q_e^2)$. 
	
	Noting that $R = 0.5/\tau_-$ and that $q_e < \tau_-$ since $e$ is non-crucial, we get $R q_e < 1$. This means that if $t_e \geq a + 1$, then $|t_e - Rq_e| \geq a$; which implies $\Pr[t_e \geq a + 1] \leq \Pr[|t_e - Rq_e| \geq a]$. Therefore by setting $a = \sqrt{R/\epsilon}$ and also using Chebyshev's inequality, we get
	\begin{equation}\label{eq:1234123489172346}
		\Pr\left[t_e \geq \sqrt{R/\epsilon}+1\right] \leq  \Pr\left[|t_e - \E[t_e]| \geq \sqrt{R/\epsilon}\right] \leq \frac{\Var[t_e]}{(\sqrt{R/\epsilon})^2} = \frac{R(q_e-q_e^2)}{(\sqrt{R/\epsilon})^2} = \epsilon (q_e - q_e^2) \leq \epsilon q_e.
	\end{equation}
	Finally, we have
	\begin{align*}
	\E\left[\frac{t_e}{R}\right] &= \underbrace{\Pr\left[\frac{t_e}{R} \leq \frac{1}{\sqrt{\epsilon R}}\right] \E\left[\frac{t_e}{R} \mid \frac{t_e}{R} \leq \frac{1}{\sqrt{\epsilon R}} \right]}_{=\E[f_e]} + \Pr\left[\frac{t_e}{R} > \frac{1}{\sqrt{\epsilon R}}\right] \underbrace{\E\left[\frac{t_e}{R} \mid \frac{t_e}{R} > \frac{1}{\sqrt{\epsilon R}} \right]}_{\leq 1 \text{ since by definition, $t_e \leq R$.}}
	\end{align*}
	Rearranging the terms and replacing the bounds specified, we get
	$$
	\E[f_e] \geq \E\left[\frac{t_e}{R}\right] - \Pr\left[\frac{t_e}{R} > \frac{1}{\sqrt{\epsilon R}}\right] = \frac{1}{R}\E\left[t_e\right] - \Pr\left[t_e \geq \sqrt{R/\epsilon} + 1 \right] \stackrel{(\ref{eq:1234123489172346})}{\geq} \frac{1}{R} \times R q_e - \epsilon q_e = (1-\epsilon)q_e,
	$$
	concluding the proof of part 2.
	
	\smparagraph{Part 3.} Note that $f_e \leq t_e/R$ by definition. Thus, we have $\sum_{e \ni v} f_e \leq \sum_{e \ni v} t_e/R = R^{-1} \sum_{e \ni v} t_e$. Since each $\MM{\mc{G}_i}$ includes at most one incident edge of $v$ for being a matching, it holds that $\sum_{e \ni v} t_e \leq R$, thus indeed $\sum_{e \ni v} f_e \leq R^{-1} R = 1$.
	
	\smparagraph{Part 4.} Let $X_i$ be the event that $v$ is matched in $\MM{\mc{G}_i}$ via a non-crucial edge and define $X := \sum_{i=1}^R X_i$. Furthermore, define for each edge $e$,
	$$
		f'_e := \begin{cases}
		\frac{t_e}{R}, & \text{if $e$ is non-crucial,}\\
		0, & \text{otherwise.}
	\end{cases}
	$$
	Note that $f'_e$ is very similar to the value of $f_e$ except for the case where $t_e/R > 1/\sqrt{\epsilon R}$. In this case, $f_e = 0$ but $f'_e$ remains to be the ratio $t_e/R$. This implies that $f'_e \geq f_e$. Now let $f'_v = \sum_{e \ni v} f'_e$. Since $f_e \leq f'_e$ for all edges, we have $f_v \leq f'_v$. Therefore, instead of proving $\Pr[f_v > n_v + 0.1\epsilon] \leq (\epsilon p)^{10}$, it suffices to prove $\Pr[f'_v > n_v + 0.1\epsilon] \leq (\epsilon p)^{10}$.

	It holds from the definition that
	$$
	f'_v = \sum_{e: e \in N, v \in e} \frac{t_e}{R} = \frac{1}{R} \sum_{e: e \in N, v \in e} t_e = \frac{1}{R} \times (X_1 + \ldots + X_R) = X/R.
	$$
	Replacing this into $\Pr[f'_v > n_v + 0.1\epsilon] \leq (\epsilon p)^{10}$, we thus have to prove 
	$
		\Pr\left[X/R > n_v + 0.1\epsilon \right] \leq (\epsilon p)^{10},
	$
	or equivalently:
	$$
		\Pr[X > R n_v + 0.1 R \epsilon] \leq (\epsilon p)^{10}.
	$$
	To prove this we use a concentration bound on $X$. Note  that the $X_i$'s are independent since graphs $\mc{G}_1, \ldots, \mc{G}_R$ are drawn independently. Moreover, for each $i \in [R]$, we have $\E[X_i] = n_v$ since recall $X_i = 1$ iff $v$ is matched via a non-crucial edge in $\MM{\mc{G}_i}$ and this has probability $\sum_{e: e\in N, v \in e} q_e = n_v$. Thus $\E[X] = Rn_v$. While we can use Chernoff's bound here since all $X_i$'s are independent, even the second-moment method is enough for our desired inequality. The variance of $X$ can be bounded as follows:
	$$
	\Var[X] = \sum_{i=1}^R \Var[X_i] = \sum_{i=1}^R E[X_i^2] - \E[X_i]^2 = R (n_v - n_v^2).
	$$
	By Chebyshev's inequality, we get
	$$
		\Pr[X > Rn_v + 0.1R\epsilon] \leq \frac{R(n_v - n_v^2)}{(0.1 R \epsilon)^2} = \frac{100(n_v - n_v^2)}{R \epsilon^2} \leq \frac{100}{R\epsilon^2}.
	$$
	Since $R = 1/2\tau_-$ and $\tau_- < (\epsilon p)^{50}$ by Corrolary~\ref{cor:thresholds}, we get 
	$$
	\Pr[X > Rn_v + R\epsilon] \leq \frac{100}{R\epsilon^2} < \frac{200(\epsilon p)^{50}}{\epsilon^2} < (\epsilon p)^{10},
	$$
	which as described above concludes the proof.
\end{proof}

\begin{proof}[Proof of Observation~\ref{obs:samedist}]
First note that realizations $\mc{C}_1, \ldots, \mc{C}_\alpha$ are all drawn precisely from the same distribution that realization $\mc{C} = \mc{C}_0$ is drawn from. Thus due to symmetry, matchings $M_0, \ldots, M_\alpha$ are all derived from the same distribution. Matchings $M'_0, \ldots, M'_\alpha$ are then the result of applying the augmenting-hyperwalks $I$ found by $\apxMIS{H, \epsilon}$ on graph $H$. Construction of graph $H$ is symmetrical w.r.t. matchings $M_0, \ldots, M_\alpha$. The only remaining component of the algorithm where this symmetry may break is in algorithm $\apxMIS{H, \epsilon}$ that may be biased towards picking augmenting-hyperwalks depending on which matching $M_i$ they would augment. This can be avoided by using an algorithm for $\apxMIS{H, \epsilon}$ that is oblivious to the indices of matchings $M_0, \ldots, M_\alpha$ used to construct graph $H$. That is, suppose e.g. that we pick the ID of nodes in $H$ randomly before feeding it into $\apxMIS{H, \epsilon}$. This guarantees that the obtained matchings $M'_0, \ldots, M'_\alpha$ will all have the same distribution due to their symmetry.
\end{proof}

\section*{Acknowledgements}

We thank Noga Alon for referring us to his paper \cite{DBLP:conf/stoc/AlonMS12} on construction of Ruzsa-Szemerédi graphs and discussing its implications which were extremely insightful. We are in addition thankful to Hamed Saleh for fruitful discussions and also to anonymous STOC reviewers for helpful suggestions.

\bibliographystyle{plain}
\bibliography{refs}

\appendix

\section{Concentration of the Maximum Realized Matching's Size}\label{sec:concentration}

In this section, we prove that random variable $\mu(\mc{G})$, i.e. the size of the maximum realized matching of $G$, is highly concentrated around its mean $\E[\mu(\mc{G})] = \opt$. A similar concentration bound was previously proved also in the works of \cite{DBLP:conf/soda/BlumCHPPV17,DBLP:conf/soda/AssadiBBMS19}. Nonetheless, we provide the full proof in this section for the sake of self-containment.

\begin{lemma}\label{lem:concentration}
	For every $0 < t \leq \opt$, $\Pr[|\mu(\mc{G}) - \opt| \geq t] \leq \exp\left(-\frac{t^2}{2\opt + 2t/3}\right) < \exp\left(-\frac{t^2}{3\opt}\right)$.
\end{lemma}

\begin{corollary}\label{cor:highprobability}
	Let $Q$ be a subgraph of $G$ obtained via a deterministic algorithm and suppose that $\opt = \omega(1)$. If $\E[\mu(\mc{Q})]/\E[\mu(\mc{G})] \geq \alpha$ then with high probability $\mu(\mc{Q})/\mu(\mc{G}) \geq (1-o(1))\alpha$.
\end{corollary}
\begin{proof}
	Lemma~\ref{lem:concentration} implies that w.h.p. $\mu(\mc{Q}) = (1\pm o(1))\E[\mu(\mc{Q})]$ and $\mu(\mc{G}) = (1\pm o(1))\E[\mu(\mc{G})]$. Therefore, w.h.p. $\mu(\mc{Q})/\mu(\mc{G}) = (1\pm o(1)) \E[\mu(\mc{Q})]/\E[\mu(\mc{G})] \geq (1-o(1))\alpha$.
\end{proof}

We note that our construction of subgraph $Q$ in Algorithm~\ref{alg:sampling} is randomized, thus the corollary above cannot be used as a black-box to imply a high probability bound. However, we remark that a similar proof to that of Lemma~\ref{lem:concentration} which we give below,  proves $\mu(\mc{Q})$ in our algorithm is concentrated around its mean even considering the randomization of Algorithm~\ref{alg:sampling}. Therefore, our algorithm also guarantees a high probability bound for the approximation-factor.

In order to prove this lemma, we use the concentration of ``self-bounding'' functions. See Sections~3.3 and 6.7 of book \cite{DBLP:books/daglib/0035704} by Boucheron, Lugosi and Massart for a thorough discussion on this concentration inequality and its proof.

\begin{definition}[{\cite[Section~6.7]{DBLP:books/daglib/0035704}}]\label{def:selfbounding}
	A function $f: \mc{X}^m \to \mathbb{R}$ is ``self-bounding'' if for every $i \in [m]$ there is a function $f_i: \mc{X}^{m-1} \to \mathbb{R}$ such that for all $x=(x_1, \ldots, x_m) \in \mathcal{X}^m$,
	\begin{enumerate}
		\item $0 \leq f(x) - f_i(x^{(i)})\leq 1$ for all $i \in [m]$, and
		\item $\sum_{i=1}^m (f(x)-f_i(x^{(i)})) \leq f(x)$,
	\end{enumerate}
	where $x^{(i)} = (x_1, \ldots, x_{i-1}, x_{i+1}, \ldots, x_n)$.
\end{definition}

\begin{lemma}[{\cite[Theorem~6.12]{DBLP:books/daglib/0035704}}]\label{lem:selfbounding}
	If $X_1, \ldots, X_m$ are independent random variables taking values in $\mathcal{X}$ and $Z = f(X_1, \ldots, X_m)$ is self-bounding, then for every $0 < t \leq \E Z$,
	$$
		\Pr[|Z - \E Z| \geq t] \leq \exp \left(- \frac{t^2}{2\E Z + 2t/3} \right).
	$$
\end{lemma}

Having this inequality, Lemma~\ref{lem:concentration} follows as follows.

\begin{proof}[Proof of Lemma~\ref{lem:concentration}]
	Let $X_e$ for each edge $e$ in graph $G$ be the indicator of the event that $e$ is realized. We can use vector $X = (X_{e_1}, \ldots, X_{e_m})$ to represent a realization of $G$ where $e_1, \ldots, e_m$ are all edges in $G$. With a slight abuse of notation, we use $\mu(X)$ to denote the size of the maximum matching in realization $X$. We first prove that function $\mu(X)$ is self-bounding. For each $i \in [m]$, define
	$$
		\mu_i(X^{(i)}) = \mu(X_{e_1}, \ldots, X_{e_{i-1}}, 0, X_{e_{i+1}}, \ldots, X_{e_m}).
	$$
	In words, $\mu_i(X^{(i)})$ is the maximum matching size in realization $X$ if we regard edge $e_i$ as unrealized. We need to show that the two conditions of Definition~\ref{def:selfbounding} hold. First, we have to show that
	$$
		0 \leq \mu(X) - \mu_i(X^{(i)}) \leq 1 \qquad \text{for all $i \in [m]$ and all realizations $X$.}
	$$
	Observe that removing a realized edge cannot increase the maximum realized matching size, thus clearly $\mu(X) - \mu(X^{(i)}) \geq 0$. Moreover, removing each edge decreases the maximum matching size by at most 1. Thus $\mu(X) - \mu(X^{(i)}) \leq 1$ proving the first condition. For the second condition, we have to show that
	$$
		\sum_{i=1}^m \left(\mu(X) - \mu_i(X^{(i)}) \right) \leq \mu(X).
	$$
	To see this, fix a maximum realized matching $M$ in realization $X$. For any edge $e_i$ outside this matching, we have $\mu(X) - \mu_i(X^{(i)}) = 0$. For the rest, as discussed above $\mu(X) - \mu_i(X^{(i)}) \leq 1$. Therefore indeed $\sum_{i=1}^m \left(\mu(X) - \mu_i(X^{(i)}) \right) \leq |M| = \mu(X)$.
	
	We proved that $\mu(X)$ is self-bounding. Since the edges are realized independently, we can plug this into Lemma~\ref{lem:selfbounding} and immediately obtain Lemma~\ref{lem:concentration}.
\end{proof}

\section{On Generality of Assumption~\ref{ass:optlarge}}\label{app:optlarge}

In this section, we prove that Assumption~\ref{ass:optlarge} comes without loss of generality. Precisely, we show that solving the problem for any input graph $G$ can be reduced to solving it for a graph $H$ with $O(\opt/\epsilon)$ vertices and $\E[\mu(\mc{H})] \geq (1-\epsilon)\opt$ where $\mc{H}$ is a realization $H$. To do this, we use a ``vertex sparsification'' idea of  Assadi~\etal{}~\cite{AKL16}. Our reduction is slightly different since we do not want parallel edges in the graph, but the main idea is essentially the same. It is also worth noting that for the reduction to work, it is crucial that our algorithm works for different edge realization probabilities. We provide the full proof for completeness.

We note that throughout the proof we may assume that $\opt$ is larger than constant $3\epsilon^{-3}$ and remark that the problem otherwise is trivial.

\smparagraph{Construction of $H$ from $G$.} We construct graph $H=(U, F)$ as follows. For $k = \frac{8\opt}{\epsilon}$, define $k$ {\em buckets} $U = \{u_1, \ldots, u_k\}$. Each of these buckets $u_i$ will correspond to a node in $H$. Assign each vertex $v$ of graph $G$ to a bucket $b(v) \in \{u_1, \ldots, u_k\}$ picked independently and uniformly at random. Then for any edge $\{v_1, v_2\}$ in graph $G$, we add an edge  $\{b(v_1), b(v_2)\}$ to $F$. Finally, we turn $H$ into a simple graph by removing self-loops and merging parallel edges.

Now we need to set the realization probability $p_e$ of every edge $e \in F$ as well. For any $e \in F$, let us denote by $E(e)$ the set of edges in the original graph $G$ that are mapped to $e$. We set
$$
	p_e := 1 - \prod_{e' \in E(e)} (1-p_{e'}).
$$
We note that $p_e$ is defined such that it precisely equals to the probability that at least one edge in $E(e)$ is realized.

\begin{claim}\label{cl:largeinH}
	Fix any matching $M$ in $G$ satisfying $|M| \leq 2\opt$. Then $\E[\mu(H)] \geq (1-\epsilon)|M|$ where the expectation is taken over the randomization of the algorithm in constructing $H$. 
\end{claim}
\begin{proof}
	Let $V(M)$ be the vertex-set of matching $M$ in graph $G$ and define 
	$$X := \{v \in V(M) \mid \exists u \in V(M) \text{ s.t. } v\not= u \text{ and } b(v)=b(u)\},
	$$
	which is the set of vertices in $V(M)$ whose bucket is not unique with regards to others in $V(M)$.
	
	We first claim that $\mu(H) \geq |M| - |X|$. Call an edge $\{u, v\} \in M$ {\em good} if $u \not\in X$, $v \not\in X$, and {\em bad} otherwise. Each bad edge has at least one endpoint in $X$, thus there are at least $|M| - |X|$ good edges in $M$. One can easily confirm that the set of corresponding edges of all good edges in $M$ forms a matching in $H$. Thus $\mu(H) \geq |M|-|X|$.
	
	To conclude, we prove that $\E[|X|] \leq \epsilon |M|$ which proves $\E[\mu(H)] \geq |M| - \epsilon |M| = (1-\epsilon)|M|$. To see why $\E[|X|] \leq \epsilon |M|$, fix any vertex $v \in V(M)$ and suppose that we have adversarially fixed the bucket $b(u)$ of all other vertices $u \in V(M)$. Since the bucket of $v$ is picked uniformly at random from $10\opt/\epsilon$ buckets and $|V(M)| \leq 2|M| \leq 4\opt$, the probability of $v$ choosing a bucket already chosen by another vertex in $V(M)$ would be $\leq \frac{4\opt}{8\opt/\epsilon} \leq \epsilon/2$. By linearity of expectation over $2|M|$ vertices in $V(M)$, we get $\E[|X|] \leq \epsilon |M|$, concluding the proof.
\end{proof}

\begin{claim}\label{cl:713713}
	It holds that $\E[\mu(\mc{H})] \geq (1-3\epsilon)\opt$. Here the expectation is taken over both the randomization in construction of $H$ and the randomization in realization $\mc{H}$ of $H$.
\end{claim}
\begin{proof}
	We first map each realization $\mc{G}$ of $G$ to a realization $\mc{H}$ of $H$. To do so, we say an edge $e \in F$ is realized in $\mc{H}$ if and only if at least one edge $e' \in E(e)$ is realized in $\mc{G}$. We argue that this mapping preserves independence of edge realizations in $H$ and their realization probabilities. First, since for any two edges $e_1, e_2 \in F$ it holds that $E(e_1) \cap E(e_2) = \emptyset$, realization of an edge $e \in F$ gives no information regarding realization of other edges. Moreover, observe that each edge $e \in F$ will be precisely realized with probability $p_e$ as discussed above in defining $p_e$.
	
	Let $M$ be the maximum realized matching of $G$. By Lemma~\ref{lem:concentration}, $\Pr[||M|-\opt| \geq \epsilon \opt] < \exp(-\frac{(\epsilon \opt)^2}{3\opt}) = \exp( - \frac{\epsilon^2 \opt}{3}) < \epsilon$ where the last inequality follows from assumption $\opt > 3\epsilon^{-3}$. This means that with probability at least $1-\epsilon$, $|M| \in [(1-\epsilon)\opt, (1+\epsilon)\opt]$. Let us suppose that this event holds and denote it by $A$. Note that event $A$ is only with regards to realization of $G$ and reveals no information about the algorithm to construct $H$. Now plugging matching $M$ into Claim~\ref{cl:largeinH}, we get that $\E[\mu(\mc{H}) \mid A] \geq (1-\epsilon)|M| \geq (1-\epsilon)(1-\epsilon)\opt \geq (1-2\epsilon)\opt$. Incorporating also the probability that event $A$ holds, which as described is at least $1-\epsilon$, we get $\E[\mu(\mc{H})] \geq (1-\epsilon)(1-2\epsilon)\opt \geq (1-3\epsilon)\opt$, concluding the proof.
\end{proof}

\smparagraph{The reduction.} We are now ready to give the full reduction. Suppose we are given $n$-vertex graph $G$ with $\opt = \E[\mu(\mc{G})]$ and assume that $\opt < 0.1 \epsilon n$ (otherwise Assumption~\ref{ass:optlarge} holds). We first construct graph $H$ as described. Note that $H$ has at most $n' = \frac{8\opt}{\epsilon}$ nodes by the construction and that $\E[\mu(\mc{H})] \geq (1-3\epsilon)\opt$ by Claim~\ref{cl:713713}. Replacing $\opt$ with $\epsilon n'/8$, we get $\E[\mu(\mc{H})] \geq (1-3\epsilon)\frac{\epsilon n'}{8}$. Assuming $\epsilon < 0.05$ (recall that we can assume $\epsilon$ to be smaller than any needed constant), this implies $\E[\mu(\mc{H})] \geq \frac{\epsilon n'}{10}$ and thus Assumption~\ref{ass:optlarge} holds for graph $H$.

Let $Q$ be the result of running Algorithm~\ref{alg:sampling} on graph $H$. Since Assumption~\ref{ass:optlarge} holds for $H$, it leads to a $(1-\epsilon)$-approximation. That is, we get $\E[\mu(\mc{Q})] \geq (1-\Omega(\epsilon))\E[\mu(\mc{H})]$. We use this subgraph $Q$ to pick a bounded-degree subgraph $Q'$ of $G$ that provides a $(1-\epsilon)$-approximation: For each edge $e \in Q$, let us {\em pick} $\min\{p^{-1} \log \epsilon^{-1}, |E(e)|\}$ arbitrary edges from $E(e)$ and put them in $Q'$. We argue that this subgraph $Q'$ has maximum degree $O_{\epsilon, p}(1)$ and that $\E[\mu(\mc{Q}')] \geq (1-\Omega(\epsilon))\opt$.

\begin{claim}
	$Q'$ has maximum degree $O_{\epsilon, p}(1)$.	
\end{claim}
\begin{proof}
	Observe that an edge $e'$ incident to a vertex $v \in V$ is in $Q'$ only if its corresponding edge $e$ in graph $H$ is in $Q$. Since $e$ corresponds to $e'$,  it should be incident to $b(v)$ of $v$ by the construction of $H$. Moreover, since $b(v)$ has maximum degree $O_{\epsilon, p}(1)$ in $Q$ and that for each edge incident to $b(v)$ in $Q$, we put at most $O(p^{-1}\log \epsilon^{-1})$ edges in $Q'$, the degree of $v$ in $Q'$ is bounded by $O_{\epsilon, p}(1) \times O(p^{-1}\log \epsilon^{-1}) = O_{\epsilon, p}(1)$. This bounds the maximum degree of $Q'$ by $O_{\epsilon, p}(1)$.
\end{proof}

\begin{claim}
	$\E[\mu(\mc{Q}')] \geq (1-\Omega(\epsilon))\opt$.
\end{claim}
\begin{proof}
	For any edge $e \in Q$, define $p'_e$ to be the probability that at least one of the edges in $G$ picked for $e$ is realized. We first argue that $p'_e \geq (1-\epsilon)p_e$. To see this, note that if $|E(e)| \leq p^{-1}\log \epsilon^{-1}$, then all the edges in $E(e)$ will be picked. Thus by definition of $p_e$ we have $p'_e = p_e$. On the other hand, if $|E(e)| > p^{-1}\log \epsilon^{-1}$, we pick exactly $p^{-1}\log \epsilon^{-1}$ edges for $e$. Since each of these edges has realization probability at least $p$, the probability that at least one of them is realized is at least
$$
	1-(1-p)^{p^{-1}\log \epsilon^{-1}} \geq 1-\epsilon \geq (1-\epsilon)p_e.
$$

Now let $M$ be any matching in $Q$. For each edge $e \in M$, choose one arbitrary edge in $E(e)$. From the construction of $H$ from $G$, one can confirm that the set of these chosen edges will form a matching of size $|M|$ in $G$. This concludes the proof: For each edge $e \in Q$, there is a probability at least $(1-\epsilon)p_e$ that one picked edge in $Q'$ is realized, thus $\E[\mu(\mc{Q}')] \geq (1-\epsilon)\E[\mu(\mc{Q})]$. As it was previously shown that $\E[\mu(\mc{Q})] \geq (1-\Omega(\epsilon))\opt$, we conclude that  $\E[\mu(\mc{Q}')] \geq (1-\Omega(\epsilon))\opt$.
\end{proof}

\section{Approximate MIS}\label{app:weakmis}

In this section we describe how Lemma~\ref{lem:apxMIS} can be derived as a corollary of the algorithm of \cite{DBLP:conf/soda/Ghaffari19}. Theorem~1.1 of \cite{DBLP:conf/soda/Ghaffari19} gives a randomized \local{} independent-set (IS) algorithm which guarantees that for each node $v$, the probability that $v$ ``has not made its decision'' after $O(\log \deg(v) + \log \frac{1}{\delta})$ rounds is at most $\delta$. The decision of $v$ is finalized if it is in the IS or it has a neighbor that is in the IS (implying that $v$ cannot be in the IS). 

To achieve Lemma~\ref{lem:apxMIS} we set $\delta = \frac{\epsilon}{10\Delta}$. Let $I$ denote the independent set returned by the algorithm after $O(\log \deg(v) + \log \frac{10\Delta}{\epsilon}) = O(\log \frac{\Delta}{\epsilon})$ rounds and let $U$ and $D$ respectively denote the set of undecided and decided vertices. We have
$$
\E[|U|] = \E\Big[ \sum_{v} \mathbbm{1}(\text{$v$ is undecided}) \Big] = \sum_v \Pr[\text{$v$ is undecided}] \leq \sum_v \frac{\epsilon}{10\Delta} = \frac{\epsilon}{10\Delta}n,
$$
and thus $\E[|D|] = n - \E[|U|] \geq (1-\frac{\epsilon}{10\Delta})n \geq 0.9n$. There is at least one IS node among the at most $\Delta + 1$ inclusive neighbors of any decided vertex; thus $\E[|I|] \geq \frac{\E[|D|]}{\Delta+1} \geq \frac{0.9n}{\Delta+1} \geq \frac{0.9n}{2\Delta} = 0.45 \frac{n}{\Delta}$. On the other hand, let $I'$ be the MIS obtained by greedily adding the undecided nodes to $I$ until they form an MIS. We have $|I'| \leq |I| + |U|$. Therefore, we indeed get that
$$
\frac{\E[|I|]}{\E[|I'|]} \geq \frac{\E[|I|]}{\E[|I|] + \E[|U|]} \geq \frac{0.45\frac{n}{\Delta}}{0.45\frac{n}{\Delta} + \frac{\epsilon}{10\Delta}n} = \frac{0.45\frac{n}{\Delta}}{(0.45 + 0.1 \epsilon) \frac{n}{\Delta}} = \frac{0.45}{0.45 + 0.1 \epsilon} > 1-\epsilon,
$$ 
concluding the proof.


\end{document}